\documentclass[acmsmall,breaklinks=true]{acmart}
\pdfoutput=1 

\usepackage[american]{babel}
\usepackage{mathtools}
\usepackage{proof}
\usepackage{mathdots}
\usepackage{bm}
\usepackage{lambdadelta2ams}
\usepackage{macros}

\newtheorem{theorem}{Theorem}[subsection]

\acmJournal{TOCL}

\title{A Formal System for the Universal Quantification of Schematic Variables}

\author{Ferruccio Guidi}
\orcid{0000-0003-3174-3248}
\affiliation{%
  \institution{University of Bologna}
  \department{Department of Computer Science and Engineering}
  \streetaddress{Mura Anteo Zamboni 7}
  \city{Bologna}
  \postcode{40126}
  \country{Italy}
}
\email{ferruccio.guidi@unibo.it}



\setcopyright{acmcopyright}

\begin{abstract}

We advocate the use of de Bruijn's universal abstraction $\TAbst{\Y}{}{}{}$
for the quantification of schematic variables
in the predicative setting
and we present a typed $\TAbst{}{}{}{}$-calculus featuring the
quantifier $\TAbst{\Y}{}{}{}$ accompanied by other practically useful
constructions like explicit substitutions and expected type annotations.
Our calculus stands just on two notions,
\ie bound rt-reduction and parametric validity,
and has the expressive power of $\LR$.
Thus, while not aiming at being a logical framework by itself,
it does enjoy many desired invariants of logical frameworks
including confluence of reduction, strong normalization, preservation
of type by reduction, decidability, correctness of types and
uniqueness of types up to conversion. 
This calculus belongs to the $\LD{}{}$ family of formal systems,
which borrow some features from the pure type systems
and some from the languages of the Automath tradition,
but stand outside both families.
In particular, our calculus includes and evolves two earlier systems
of this family. Moreover, a machine-checked specification of its theory
is available.

\end{abstract}

\begin{document}

\maketitle


\section{Introduction}
\seclabel{introduction}

Mathematical theories make frequent use of variables that are universally
quantified in the meta-language and that are known as schematic variables.
We see them especially in the predicative or constructive setting,
where impredicative or second-order quantification
must be avoided in the language of the theories.
Taking an informal example,
Hilbert-style intuitionistic logic contains the axiom:
$\MATOM{\phi \to \psi \to \phi \land \psi}$,
in which the variables $\phi$ and $\psi$ are necessarily schematic,
as quantifying them in the axiom, would be impredicative and thus
contrary to the conventions of intuitionism.

On the other hand, in the meta-language of this axiom less restrictions apply
and we can quantify over $\phi$ and $\psi$ even if they are
second-order variables, but the quantification must be predicative.

This is indeed the case when the meta-language takes the form of
a framework enforcing the \emph{propositions-as-types} (PAT) interpretation
\cite{SU06}. In fact $\phi$ and $\psi$, that stand for propositions,
range in this framework over types.
Thus, in order to be adequate, the framework must support
second-order universal quantification in a predicative form.

The literature on typed $\lambda$-calculus
addresses this issue in more than one way.
In the world of pure type systems (PTS) \cite{Brn92}
the quantified axiom $\MALL{\phi}\MALL{\psi}\MATOM{\phi \to \psi \to \phi \land \psi}$
and its instances can reside in different sorts
(for instance, $\SortC$ and $\SortA$ respectively
in $\LQE$ [\ibid{}, example 5.2.4.8]).
With the refined pure type systems of \citeN{KLN03},
the variables $\phi$ and $\psi$ can be parameters,
thus the quantified axiom can take the form:
$\MPAR{\phi}\MPAR{\psi}\MATOM{\phi \to \psi \to \phi \land \psi}$.

In the frameworks of the Automath tradition \cite{SPA94},
the schematic variables are quantified with de Bruijn's abstraction,
denoted hereafter by $\TAbst{\Y}{}{}$.
As the $\TAbst{}{}{}$-abstraction of a PTS,
the $\TAbst{\Y}{}{}$-abstraction is a weak head normal form and undergoes $\beta$-reduction,
but the type of a $\TAbst{\Y}{}{}$-abstraction is itself a $\TAbst{\Y}{}{}$-abstraction.
Hence a $\TAbst{\Y}{}{}$-abstraction occurring at the level of types
acts as a predicative quantification.
More precisely, $\CType{\G}{\TAbst{\Y}{x}{W}T}{W}$ is false for every term $W$.

So, informally, the quantified axiom would take the form:
$\MABST{\phi}\MABST{\psi}\MATOM{\phi \to \psi \to \phi \land \psi}$.

In this article we wish to study the $\TAbst{\Y}{}{}$-abstraction
by building a rather minimal formal system around it.
This system does not aim at being a logical framework,
but just a fragment of some wider framework to be developed in the future.
In this respect, our present goal is just to assert that the system enjoys
many desirable properties already in itself
and to show how these properties must be stated to produce
the inductive hypotheses allowing to prove them.
The reader wishing to glance a conjectured logical framework
that includes our system may look at \citeN{lambdadeltaJ3a}.

We define our system in \secref{definition} and show it enjoys
the three main invariants referred to as \emph{the three problems} in Automath literature.
In particular, confluence is in \secref{transition},
strong normalization is in \secref{normalization}
and preservation, also known as subject reduction, is in \secref{preservation}.

We remark that our system is typed,
but type assignment is not a primitive notion.
So the reader will find our type rules in \secref{types} and \secref{ntas}
where the properties of types are proved.

The reader will also find our version of
the so-called big-tree theorem, stated first by \citeN{SPAc4},
and the confluence of bound rt-reduction, which we are presenting here
for the first time.

Stating and proving such properties requires quite a number of auxiliary
notions that we define in the main text and in the appendices.
One of the most relevant is the arity assignment of \secref{lsubc}.
Using arities, also known as norms or skeletons,
we can show that our system has the expressive power of $\LR$.
And yet, $\TAbst{\Y}{}{}$-abstraction improves $\LR$ by permitting uniformly dependent types.
This means that types can contain universally quantified variables,
but each variable can be instantiated just with terms having a fixed arity
determined in advance by its quantifier.

Our proofs passed a full computer check,
so we just outline them in the article
reporting on proof strategies and on main dependencies.
Most proofs are broken in many easy cases which we omit the details of.
Proofs by cases on a premise are by cases on the last step of its derivation.
Proofs by induction on a premise are by induction on 
the length of its derivation and are by cases as well.
Proofs by induction on a closure are by well-founded induction on the
proper subclosure relation.
Proofs by big-tree induction are by well-founded induction
on proper qrst-reduction of \secref{lsubv}.

The article is organized for the reader willing to see
the definitions and the propositions at first,
so the lengthy explanations on the system's features are in \secref{conclusion}
with our concluding remarks, where we highlight our contribution,
we discuss related work and we outline future perspectives.

Most concepts and methods discussed in this article
are not a novelty in the world of typed $\lambda$-calculus.
However, we would like to stress that the next two ideas seem original to us.

Bound rt-reduction (\secref{cpms})
combines any number of reduction steps and a given number of type inference steps
in a single relation.
A similar notion, allowing any number of type inference steps,
appears in the Automath tradition, and here in \secref{cpxs},
as a marginal device introduced just to prove specific theorems.
On the contrary, we advocate the central role of bound rt-reduction
by making it one of the pillars on which we build the system we are presenting.

Sort irrelevance (\secref {feqx}) equates terms that differ just
in the sorts they contain.
With the help of this equivalence relation,
that to our knowledge does not appear in the literature,
we can prove the big-tree theorem without the degree-based induction
on which its original proof is based.
At this point the notion of degree, recalled for convenience
in \secref{abstraction}, and its related theory can be fully removed
from the system's presentation in front of the reader.

\section{Definition of the System}
\seclabel{definition}

Our formal system,
whose syntax we explain in \secref{syntax},
stands on two notions:
bound rt-reduction (\secref{cpms}) and
parametric validity (\secref{nta}).
The first one comprises both ordinary reduction
and inferred type assignment.
The second one includes expected type checking.

In this article we borrow the distinction between
an expected type and an inferred type from \citeN{Cos96}.
Thus, an inferred type stands for what is also known as a synthesized type.

We recall that expected types are given in advance whereas inferred
types are computed from terms by type inference rules.
In this respect an inferred $2$-type of a term is an inferred type of an inferred type
of that term and in general we will consider inferred $n$-types for any $n \ge 0$.

Moreover, the expected types of a term form an equivalence class with
respect to conversion and the inferred types (usually just one)
of that term are specific representatives in the class.

\subsection{Syntactic Categories}
\seclabel{syntax}

Our grammar features two syntactic categories:
the terms and the environments of \tabref{tl} explained next.
In terms
$\TSRef{s}$ denotes
a generic sort ranging over a set $\Sort$ with at least one element.
$\TNRef{x}$ denotes
a generic variable ranging over a countable set $\SVar$.
$\TAppl{V}T$ denotes
the application of the function $T$ to the argument $V$,
displayed according to the so-called item notation of \citeN{KN96b}
in order to improve the visual understanding of redexes.
$\TAbst{\Y}{x}{W}T$ denotes
de Bruijn's abstraction in $T$ of $\TNRef{x}$ with expected type $W$.
$\TAbbr{x}{V}T$ denotes
the explicit substitution in $T$ of $\TNRef{x}$ with $V$.
$\TCast{U}T$ denotes
the annotation of $T$ with its expected type $U$ (also known as a type cast). 
Moreover in environments
$\LAtom$ denotes the empty environment.
$K\LAbst{x}{W}$ and $K\LAbbr{x}{V}$ denote
the declaration of $\TNRef{x}$ with expected type $W$
and
the definition of $\TNRef{x}$ as $V$,
both in the environment $K$.

\begin{table}
\seccaption{syntax}
{Terms and environments.}
\tablabel{tl}
\begin{tabular}{lr@{\;}l}
Term:& 
$T,U,V,W,X \GDEF$&
$
\TSRef{s} \GOR
\TNRef{x} \GOR
\TAppl{V}T \GOR
\TAbst{\Y}{x}{W}T \GOR
\TAbbr{x}{V}T \GOR
\TCast{U}T
$\nl
Environment:&
$K,L \GDEF$&
$
\LAtom \GOR
K\LAbst{x}{W} \GOR
K\LAbbr{x}{V}
$\\ 
\end{tabular}
\end{table}

Following a widely accepted convention,
meta-variables for names will consist of lowercase letters,
\ie usually $\TNRef{x}$, whereas uppercase letters
will denote meta-variables for terms and other entities.
The symbol $\TAppl{}$ (circled letter \emph{a}) is widely used for application,
the symbol $\TAbst{}{}{}$ is standard for functional abstraction,
the symbol $\TAbbr{}{}$ is taken after $\delta$-expansion,
\ie the operation of unfolding a definition or an explicit substitution, 
the symbol $\TCast{}$ (circled letter \emph{c}) is taken after \emph{cast}.

Our grammar uses the items $\TAbst{\Y}{x}{W}$ and $\TAbbr{x}{V}$
both in environments, where we will term them entries,
and in terms to reduce the number of displayed notations
and to suggest the embedding of environments in terms we envisioned in
\citeN{lambdadeltaJ1a}.
In addition, we introduce the shared notations of \tabref{shared}
to reduce the number of displayed rules.
If these notations occur more than once in a rule or in a statement,
they have the same meaning in every occurrence.

\begin{table}
\seccaption{syntax}
{Shared notations.}
\tablabel{shared}
\begin{tabular}{ll}
$\TPair{\Y}{x}{V}$ means $\TAbst{\Y}{x}{V}$ or $\TAbbr{x}{V}$&
$\TFlat{V}$ means $\TCast{V}$ or $\TAppl{V}$\\
\end{tabular}
\end{table}

Finally, notice that bound and free variable occurrences are
implicitly defined as one expects.
Moreover, we assume the variable convention of \citeN{Brn92},
\ie in every mathematical context
all bound variables are distinct and, moreover, they are different from all free variables.

We recall that the scope of binders
extends as much as possible at their right respecting parentheses.
This convention holds for terms binders and for environment entries as well.
Moreover, the entries of an environment $L$ bind the free variables
of a term $T$ in a closure $\Cl{L}{T}$ (see \secref{fqus}). 

A term $T$ is closed in $L$ when $L$ binds all free variables of $T$. 
The entries of $L$ recursively referred by $T$ are
the entries of $L$ binding the free variables of $T$
and the entries of $L$ recursively referred by the terms in these entries.
The inherited subterms of $T$ in $L$ are
the subterms of $T$ (including the term itself)
and the subterms of the terms in the entries of $L$ recursively referred by $T$.

\subsection{Reduction and Type Inference for Terms}
\seclabel{cpms}

In this section we define reduction and type inference at once
by means of bound rt-reduction, which we introduce in this article for
the first time.
This reduction system is environment-aware and deterministic.
In order to ease the proof of its confluence,
it minimizes the number of critical pairs and avoids the replication of residual redexes. 
It is derived from the reduction system used by \citeN{SPAc4}
to prove the strong normalization of $\VrLL$
and comprises two families of steps in addition to context rules.
The r-steps are for small-step reduction and
the t-steps are for basic type inference.
So, hereafter, r-reduction will comprise r-steps and context rules only,
\ie it will refer to ordinary reduction, while t-reduction will
comprise t-steps and context rules only.

The relation $\CPM{n}{L}{T_1}{T_2}$ defined in \tabref{cpm} denotes
one step of bound rt-reduction from $T_1$ to $T_2$ in $L$.
Here we can prove $n \in \SUBSET{0,1}{}$
and the two numbers have the next meaning.
If $n = 0$ then $T_2$ is an r-reduct of $T_1$,
\ie a reduct in the ordinary sense.
If $n = 1$ then $T_2$ is an r-reduct of an inferred type of $T_1$.
As a matter of fact, we do not need to define inferred types in themselves.
We just need to define their reducts. Information on how to define
inferred types is in \appref{cpts}.
The parameter $\Next{}$ appearing in Rule $\ruleref{cpm}{s}$
is a function from $\Sort$ to $\Sort$ that can be chosen at will.

\begin{table}
\seccaption{cpms}
{Bound rt-reduction for terms (one step).}
\tablabel{cpm}
\begin{tabular}{c}

Reduction rules (r-steps)\nl

\infer[\ruleref{cpm}{\beta}]
{\CPM{0}{L}{\TAppl{V}\TAbst{\Y}{x}{W}T}{\TAbbr{x}{\TCast{W}V}T}}
{}
\sep

\infer[\ruleref{cpm}{\delta}]
{\CPM{0}{K\LAbbr{x}{V}}{\TNRef{x}}{V}}
{}
\sep

\infer[\ruleref{cpm}{\zeta}]
{\CPM{0}{L}{\TAbbr{x}{V}T}{T}}
{}
\nl

\infer[\ruleref{cpm}{\theta}]
{\CPM{0}{L}{\TAppl{V}\TAbbr{x}{W}T}{\TAbbr{x}{W}\TAppl{V}T}}
{}
\sep

\infer[\ruleref{cpm}{\epsilon}]
{\CPM{0}{L}{\TCast{U}T}{T}}
{}
\nl

Type inference rules (t-steps)\nl

\infer[\ruleref{cpm}{s}]
{\CPM{1}{L}{\TSRef{s}}{\TSRef{\Next{s}}}}
{}
\sep

\infer[\ruleref{cpm}{l}]
{\CPM{1}{K\LAbst{x}{W}}{\TNRef{x}}{W}}
{}
\sep

\infer[\ruleref{cpm}{e}]
{\CPM{1}{L}{\TCast{U}T}{U}}
{}
\nl

Context rules\nl

\infer[\ruleref{cpm}{L}]
{\CPM{n}{K\LPair{y}{V}}{\TNRef{x}}{T}}
{\CPM{n}{K}{\TNRef{x}}{T}}
\sep

\infer[\ruleref{cpm}{\TAppl{}l}]
{\CPM{0}{L}{\TAppl{V_1}T}{\TAppl{V_2}T}}
{\CPM{0}{L}{V_1}{V_2}}
\sep

\infer[\ruleref{cpm}{\TAppl{}r}]
{\CPM{n}{L}{\TAppl{V}T_1}{\TAppl{V}T_2}}
{\CPM{n}{L}{T_1}{T_2}}
\nl

\infer[\ruleref{cpm}{Pl}]
{\CPM{0}{L}{\TPair{\Y}{x}{V_1}T}{\TPair{\Y}{x}{V_2}T}}
{\CPM{0}{L}{V_1}{V_2}}
\sep

\infer[\ruleref{cpm}{Pr}]
{\CPM{n}{L}{\TPair{\Y}{x}{V}T_1}{\TPair{\Y}{x}{V}T_2}}
{\CPM{n}{L\LPair{x}{V}}{T_1}{T_2}}
\nl

\infer[\ruleref{cpm}{\TCast{}l}]
{\CPM{0}{L}{\TCast{U_1}T}{\TCast{U_2}T}}
{\CPM{0}{L}{U_1}{U_2}}
\sep

\infer[\ruleref{cpm}{\TCast{}r}]
{\CPM{0}{L}{\TCast{U}T_1}{\TCast{U}T_2}}
{\CPM{0}{L}{T_1}{T_2}}
\sep

\infer[\ruleref{cpm}{\TCast{}b}]
{\CPM{1}{L}{\TCast{U_1}T_1}{\TCast{U_2}T_2}}
{\CPM{1}{L}{U_1}{U_2}&\CPM{1}{L}{T_1}{T_2}}
\nl

Rule $\ruleref{cpm}{\zeta}$:
$x$ not free in $T$.
\sep
Rule $\ruleref{cpm}{L}$:
$y \neq x$ and $y$ not free in $T$.
\\

\end{tabular}
\end{table}

Five Greek letters designate the r-steps:
$\beta$ is standard for function application,
$\delta$ is widely used for unfolding a definition,
$\zeta$ is widely used for removing a not referred definition,
$\theta$ appears in \citeN{CH00} for swapping an application-definition pair,
$\epsilon$, being alphabetically near to $\zeta$, is used here for removing a type annotation.
Three Latin letters designate the t-steps:
$s$, taken after \emph{sort}, is used here to type a sort,
$l$, taken after $\lambda$, is used here to type a declared variable
occurrence by applying the pattern of $\delta$-expansion to a $\TAbst{\Y}{}{}$-item,
$e$, taken after the r-step $\epsilon$, is used here for typing a term of which
the expected type is given by an annotation item.

According to what we stated in \secref{introduction},
the rt-reduction rules show that a $\TAbst{\Y}{}{}$-abstraction
is a weak head normal form.
Rule $\ruleref{cpm}{\beta}$ shows it undergoes $\beta$-reduction
(in delayed form, \ie introducing an explicit substitution)
and Rule $\ruleref{cpm}{Pr}$ for $n=1$ shows that its inferred type
is a $\TAbst{\Y}{}{}$-abstraction.

The usual big-step $\beta$-contraction is realized by
one small-step $\beta$-contraction followed by
zero ore more $\delta$-expansions
and one $\zeta$-contraction.
The $e$-contraction and the expected type $W$
in the $\beta$-reduct
are the key devices to extend without effort
the strong normalization of r-steps alone
to the strong normalization of r-steps and t-steps combined.
The $\theta$-swap is meant to realize
the $\beta$-contraction at a distance, \ie the $\beta_1$-step of \citeN{Ned73}.
Indeed, this notion is convenient in many situations as
the ones explored by \citeN{AGLK18,AL16}. 

The r-steps and the t-steps do not commute in general and
$(\epsilon, e)$ is a critical pair,
which is confluent if the redex $\TCast{U}T$ is valid.
On the other hand, r-reduction alone is always confluent.

Finally, the relation $\CPMS{n}{L}{T_1}{T_2}$
defined in \tabref{cpms} as 
the reflexive and transitive closure of $\CPM{n}{L}{T_1}{T_2}$,
denotes a bound rt-reduction sequence, \ie a computation,
from $T_1$ to $T_2$ in $L$.

\begin{table}
\seccaption{cpms}
{Bound rt-reduction for terms (sequence of steps).}
\tablabel{cpms}
\begin{tabular}{c}

\infer[\ruleref{cpms}{R}]
{\CPMS{0}{L}{T}{T}}
{}
\sep

\infer[\ruleref{cpms}{I}]
{\CPMS{n}{L}{T_1}{T_2}}
{\CPM{n}{L}{T_1}{T_2}}
\sep

\infer[\ruleref{cpms}{T}]
{\CPMS{n_1+n_2}{L}{T_1}{T_2}}
{\CPMS{n_1}{L}{T_1}{T}&\CPMS{n_2}{L}{T}{T_2}}
\\

\end{tabular}
\end{table}

Here $n$ can be any non-negative integer and
$T_2$ is an r-reduct of an inferred $n$-type of $T_1$.

We will say that two terms r-convert in $L$ when they have a common
r-reduct in $L$.

\subsection{Validity and Type Checking for Terms}
\seclabel{nta}

Generally speaking, a validity condition discriminates
the $\lambda$-terms having specific properties.

In our system, the validity predicate $\CNV{\A}{L}{X}$
we define with the rules of \tabref{cnv} 
asserts the next conditions \itref{closed} to \itref{strong}
for the term $X$ in the environment $L$.
The parameter $\A$ appearing in Rule $\ruleref{cnv}{\TAppl{}}$
and termed hereafter the applicability domain
is a subset of integers that can be chosen at will.

\begin{enumerate}

\item
\itlabel{closed}
The term $X$ is closed in $L$;

\item
\itlabel{cast}
if $X$ is $\TCast{U}{T}$,
an inferred type of $T$ r-converts to the expected type $U$ in $L$; 

\item
\itlabel{appl}
if $X$ is $\TAppl{V}{T}$,
an inferred $n$-type of $T$
r-converts to a functional term $\TAbst{\Y}{x}{W}U$ in $L$
and an inferred type of $V$ r-converts to the expected type $W$ in $L$,
provided that $n \in \A$;

\item
\itlabel{strong}
every inherited subterm of $X$ in $L$ is valid in $L$
(this property is known as strong validity).

\end{enumerate}

The distinction between strong and weak validity is due to \citeN{Bru68b}.
The reader will find more remarks in \secref{application}
where we show a term that is weakly (but not strongly) valid. 

\begin{table}
\seccaption{nta}
{Validity for terms.}
\tablabel{cnv}
\begin{tabular}{c}

\infer[\ruleref{cnv}{\TSort{}}]
{\CNV{\A}{L}{\TSRef{s}}}
{}
\sep

\infer[\ruleref{cnv}{\TLRef{}}]
{\CNV{\A}{K\LPair{x}{V}}{\TNRef{x}}}
{\CNV{\A}{K}{V}}
\sep

\infer[\ruleref{cnv}{L}]
{\CNV{\A}{K\LPair{y}{V}}{\TNRef{x}}}
{\CNV{\A}{K}{\TNRef{x}}}
\nl

\infer[\ruleref{cnv}{P}]
{\CNV{\A}{L}{\TPair{\Y}{x}{V}T}}
{\CNV{\A}{L}{V}&\CNV{\A}{K\LPair{x}{V}}{T}}
\sep

\infer[\ruleref{cnv}{\TCast{}}]
{\CNV{\A}{L}{\TCast{U}T}}
{\CNV{\A}{L}{U}&\CNV{\A}{L}{T}&\CPMS{1}{L}{T}{U_0}&\CPMS{0}{L}{U}{U_0}}
\nl

\infer[\ruleref{cnv}{\TAppl{}}]
{\CNV{\A}{L}{\TAppl{V}T}}
{\CNV{\A}{L}{V}&\CNV{\A}{L}{T}&n\in\A&\CPMS{n}{L}{T}{\TAbst{\Y}{x}{W_0}U_0}&\CPMS{1}{L}{V}{W_0}}
\nl

Rule $\ruleref{cnv}{L}$:
$y \neq x$.
\\

\end{tabular}
\end{table}

Typical choices for the applicability domain emerging from the
Automath tradition are $\A \defeq \ACAny$ or $\A \defeq \ACOne$,
but we are aiming to develop our system without fixing $\A$ in advance. 

In the end
we define our type judgment $\NTA{\A}{L}{T}{U}$
in \tabref{nta} as a valid type annotation.

\begin{table}
\seccaption{nta}
{Type checking for terms.}
\tablabel{nta}
\begin{tabular}{c}

\infer[\ruleref{nta}{\TCast{}}]
{\NTA{\A}{L}{T}{U}}
{\CNV{\A}{L}{\TCast{U}T}}
\\

\end{tabular}
\end{table}

Contrary to other frameworks like the $\LCube$,
where the sort $\SortB$ is valid but not typable,
we will see that in our system validity corresponds to typability.
Our validity is a primitive notion mainly because
preservation by r-reduction is easier to prove
for validity than for type assignment
when the parameter $\A$ is not fixed in advance,
as types are assigned up to r-conversion. 
Notice that \citeN{KN96a} define typability
as a primitive notion in the $\LCube$ as well.

The reader will find the axioms of our type judgment in \secref{types}, 
\secref{ntas} and \appref{members}.

Given a term $T$, we remark that an r-reduct $U$ of an inferred type of $T$
in the sense of $\CPMS{1}{L}{T}{U}$
is not in general a type of $T$ in the sense of $\NTA{\A}{L}{T}{U}$.
Strictly speaking, $U$ is just a \emph{pretype} of $T$.
As of \thref{types}{cnv_cpms_nta} of \secref{types},
the term $U$ is indeed a type of $T$ when $T$ itself is valid.

\section{QRST-Reduction and Related Notions}
\seclabel{transition}

Having defined our system in \secref{definition}, 
we devote the remaining sections to the main results of the system's meta-theory.
In this respect our first objective is to reach
our version of the so-called big-tree theorem,
that was first stated by \citeN{SPAc4} and that we will state in
\secref{normalization}. 
This result states the strong normalization of a relation,
introduced in \secref{fpbs} and termed hereafter qrst-reduction,
that generalizes
both the bound rt-reduction relation (\secref{cpxs}, \secref{lpxs})
and the subterm relation (\secref{fqus}).
The distinctive features of qrst-reduction are the generalized t-step of
Rule $\ruleref{cpx}{s}$ and the q-step that we introduce in \secref{req}.
Our development requires to define a couple of auxiliary notions
(\secref{lsubr} and \secref{freep}) that give us the opportunity to
prove the confluence of r-reduction, \ie its Church-Rosser property,
in \secref{lprs}.
The other main result is \thref{fpbs}{fpbs_inv_star} (\secref{fpbs})
about the decomposition of qrst-reduction sequences.

\subsection{Structural Order for Closures}
\seclabel{fqus}

Following \citeN{Ned73},
we observe that, contrary to untyped $\lambda$-calculus,
typed $\lambda$-calculus is concerned just with valid terms,
which are not absolutely open.
Thus, the notion of closure, \ie a pair $\Cl{L}{T}$ of an
environment $L$ binding the free variables of a term $T$,
emerges naturally. 

Induction on the structure of closures is granted by the
strict partial order emerging from the direct subclosure relation
$\FQU{L_1}{T_1}{L_2}{T_2}$ that we define in \tabref{fqu}
and that represents the s-step of our qrst-reduction system.
Its transitive closure $\FQUP{L_1}{T_1}{L_2}{T_2}$,
\ie the proper subclosure relation, is defined in \tabref{fqup}.
Also the reflexive and transitive closure $\FQUS{L_1}{T_1}{L_2}{T_2}$
defined by the rules of \tabref{fqus} is interesting 
as it appears in the decomposition \thref{fpbs}{fpbs_inv_star}.

\begin{table}
\seccaption{fqus}
{Direct subclosure (one s-step).}
\tablabel{fqu}
\begin{tabular}{c}

\infer[\ruleref{fqu}{\TLRef{}}]
{\FQU{K\LPair{x}{V}}{\TNRef{x}}{K}{V}}
{}
\sep

\infer[\ruleref{fqu}{L}]
{\FQU{K\LPair{y}{V}}{T}{K}{T}}
{}
\sep

\infer[\ruleref{fqu}{Fl}]
{\FQU{L}{\TFlat{V}T}{L}{V}}
{}
\nl

\infer[\ruleref{fqu}{Pl}]
{\FQU{L}{\TPair{\Y}{x}{V}T}{L}{V}}
{}
\sep

\infer[\ruleref{fqu}{Pr}]
{\FQU{L}{\TPair{\Y}{x}{V}T}{L\LPair{x}{V}}{T}}
{}
\sep

\infer[\ruleref{fqu}{Fr}]
{\FQU{L}{\TFlat{V}T}{L}{T}}
{}
\nl

Rule $\ruleref{fqu}{L}$:
$y$ not free in $T$.\\

\end{tabular}
\end{table}

By observing that s-steps decrease the sum
of the term constructors occurring in closures,
we easily argue that infinite sequences of s-steps are impossible.
Using other words,
we are saying that the s-reduction system is strongly normalizing
and that $\FQUP{L_1}{T_1}{L_2}{T_2}$ is well-founded.

We want to stress that
restricting $T$ to $x$ in Rule $\ruleref{fqu}{L}$,
as we do in the other $L$-rules
(see for instance Rule $\ruleref{cpm}{L}$ and Rule $\ruleref{cnv}{L}$)
invalidates the important \thref{lpxs}{fqu_cpx_trans}
presented in \secref{lpxs}.

\subsection{Extended RT-Reduction for Terms}
\seclabel{cpxs}

Extended rt-reduction is the union of the reduction relation
and the inferred-type relation.
It emerges by removing the bounds on the number of t-steps
in the definition of bound rt-reduction (\secref{cpms}).
Thus, we define in \tabref{cpx} 
one extended rt-step for terms with the relation $\CPX{L}{T_1}{T_2}$
that, incidentally, does not depend on the parameter $\Next{}$ since
we generalize Rule $\ruleref{cpm}{s}$ as follows:
\begin{equation}
\eqnlabel{cpx_s}
\infer[\ruleref{cpx}{s}]
{\CPX{L}{\TSRef{s_1}}{\TSRef{s_2}}}
{}
\end{equation}
Not surprisingly, we will also need $\CPXS{L}{T_1}{T_2}$, 
\ie the reflexive and transitive closure of the previous relation,
that we define in \tabref{cpxs} as one expects.

\subsection{Refinement for the Preservation of RT-Reduction}
\seclabel{lsubr}

The theory of our system contains
reflexive relations for environments termed here refinements.
They are invoked when proving that rt-reduction preserves some property
and, specifically, they come into play in the case of the $\beta$-rule,
given that the backward application of the corresponding P-rule
moves part of the $\beta$-redex and part of the $\beta$-reduct
in the environment.

The rules of \tabref{lsubr} define 
the relation $\LSubR{L_1}{L_2}$
stating that $L_1$ refines $L_2$ for the preservation of rt-reduction.
This refinement is transitive (see \citeN{lambdadeltaR2c})
and is implied by the other refinements.

\begin{theorem}[refinement for the preservation of rt-reduction]\
\thslabel{lsubr}
\begin{enumerate}

\item\thlabel{lsubr_cpm_trans}
\Caption{strengthening of bound rt-reduction through refinement}
If $\LSubR{K}{L}$
and $\CPM{n}{L}{T_1}{T_2}$
then $\CPM{n}{K}{T_1}{T_2}$.

\item\thlabel{lsubr_cpx_trans}
\Caption{strengthening of extended rt-reduction through refinement}
If $\LSubR{K}{L}$
and $\CPX{L}{T_1}{T_2}$
then $\CPX{K}{T_1}{T_2}$.

\end{enumerate}
\end{theorem}

\begin{proof}
\thref{}{lsubr_cpm_trans}
is proved by induction on the second premise
and by cases on the first premise.
Rule $\ruleref{cpm}{e}$ is essential
in the case of Rule $\ruleref{cpm}{l}$ against Rule $\ruleref{lsubr}{\beta}$.
\thref{}{lsubr_cpx_trans}
is proved as \thref{}{lsubr_cpm_trans}.
\end{proof}

\begin{table}
\seccaption{lsubr}
{Refinement for the preservation of rt-reduction.}
\tablabel{lsubr}
\begin{tabular}{c}

\infer[\ruleref{lsubr}{\LAtom}]
{\LSubR{\LAtom}{\LAtom}}
{}
\sep

\infer[\ruleref{lsubr}{B}]
{\LSubR{K_1\LPair{y}{V}}{K_2\LPair{y}{V}}}
{\LSubR{K_1}{K_2}}
\sep

\infer[\ruleref{lsubr}{\beta}]
{\LSubR{K_1\LAbbr{y}{\TCast{W}V}}{K_2\LAbst{y}{W}}}
{\LSubR{K_1}{K_2}}
\\

\end{tabular}
\end{table}

\subsection{R-Reduction for Environments}
\seclabel{lprs}

R-reduction for environments
is a standard tool in the development of typed $\lambda$-calculus.
For instance \citeN{Brn92} uses it to prove the preservation theorem,
\ie subject reduction, for a PTS
(even if \citeN{BW97} and other authors can avoid it).
In our system we use it more often
since most of our relations are aware of environments
and these contain explicit substitutions.

The rules of \tabref{lpr} define the relation $\LPR{L_1}{L_2}$
denoting one step of r-reduction for environments.
And \tabref{lprs} defines its reflexive and transitive closure $\LPRS{L_1}{L_2}$
in the usual way.

\begin{table}
\seccaption{lprs}
{R-reduction for environments (one step on all entries).}
\tablabel{lpr}
\begin{tabular}{c}

\infer[\ruleref{lpr}{B}]
{\LPR{K_1\LPair{y}{V}}{K_2\LPair{y}{V}}}
{\LPR{K_1}{K_2}}
\sep

\infer[\ruleref{lpr}{P}]
{\LPR{K\LPair{y}{V_1}}{K\LPair{y}{V_2}}}
{\CPM{0}{K}{V_1}{V_2}}
\\

\end{tabular}
\end{table}

Notice that it makes no sense to define these notions for a bound $n \neq 0$
since the premise of Rules $\ruleref{cpm}{\TAppl{}l}$ and $\ruleref{cpm}{Pl}$
is $\CPM{0}{L}{V_1}{V_2}$.
So the bound appears just for uniformity.

R-reduction for environments comes into play considering confluence
of r-reduction for terms.

\begin{theorem}[confluence of r-reduction]\
\thslabel{lprs}
\begin{enumerate}

\item\thlabel{lpr_cpm_trans}
\Caption{transitivity of r-reduction on all entries and bound rt-reduction for terms}
If $\LPR{K}{L}$
and $\CPM{n}{L}{T_1}{T_2}$
then $\CPMS{n}{K}{T_1}{T_2}$.

\item\thlabel{cpr_conf_lpr}
\Caption{diamond confluence of r-reduction for terms with itself}
If $\CPM{0}{L_0}{T_0}{T_1}$
and $\CPM{0}{L_0}{T_0}{T_2}$
and $\LPR{L_0}{L_1}$
and $\LPR{L_0}{L_2}$
with $T_1 \neq T_2$\\
then there exists $T$
such that $\CPM{0}{L_1}{T_1}{T}$
and $\CPM{0}{L_2}{T_2}{T}$.

\item\thlabel{cprs_conf}
\Caption{full confluence of r-reduction for terms with itself}
If $\CPMS{0}{L}{T_0}{T_1}$
and $\CPMS{0}{L}{T_0}{T_2}$
then there exists $T$
such that $\CPMS{0}{L}{T_1}{T}$
and $\CPMS{0}{L}{T_2}{T}$.

\item\thlabel{lpr_conf}
\Caption{diamond confluence of r-reduction on all entries with itself}
If $\LPR{L_0}{L_1}$
and $\LPR{L_0}{L_2}$
with $L_1 \neq L_2$
then there exists $L$
such that $\LPR{L_1}{L}$
and $\LPR{L_2}{L}$.

\item\thlabel{lprs_conf}
\Caption{full confluence of r-reduction on all entries with itself}
If $\LPRS{L_0}{L_1}$
and $\LPRS{L_0}{L_2}$
then there exists $L$
such that $\LPRS{L_1}{L}$
and $\LPRS{L_2}{L}$.

\end{enumerate}
\end{theorem}

\begin{proof}
\thref{}{lpr_cpm_trans}
is proved by induction on the second premise
and by cases on the first premise.
\thref{}{cpr_conf_lpr}
is proved by induction on the closure $\Cl{L_0}{T_0}$
and by cases on the premises.
Notice that we need \thref{lsubr}{lsubr_cpm_trans}
in the cases of Rule $\ruleref{cpm}{\beta}$.
\thref{}{lpr_conf} follows from \thref{}{cpr_conf_lpr}.
\thref{}{cprs_conf} (the confluence theorem) and \thref{}{lprs_conf}
follow from \thref{}{cpr_conf_lpr} and \thref{}{lpr_conf} respectively
by invoking strip lemmas after \citeN{Brn92}.
Notice that \thref{}{lprs_conf} needs \thref{}{lpr_cpm_trans} as well.
\end{proof}

Notice as well that $\LPRS{L_1}{L_2}$ has the next alternative axiomatization
because of \thref{lprs}{lpr_cpm_trans}.
\begin{equation}
\eqnlabel{lprs_alt}
\vcenter{
\infer[\ruleref{lprs}{\LAtom}]
{\LPRS{\LAtom}{\LAtom}}
{}
}\sep
\vcenter{
\infer[\ruleref{lprs}{P}]
{\LPRS{K_1\LPair{y}{V_1}}{K_2\LPair{y}{V_2}}}
{\LPRS{K_1}{K_2}&\CPMS{0}{K_1}{V_1}{V_2}}
}
\end{equation}

\subsection{Extended RT-Reduction for Environments}
\seclabel{lpxs}

The extended rt-steps are meant to act on closures as the s-steps of
\secref{fqus}, so we need to define them for environments as well as for terms.
We do it with the relation $\LPX{L_1}{L_2}$ of \tabref{lpx} and with
its reflexive and transitive closure $\LPXS{L_1}{L_2}$ of \tabref{lpxs}.
The two emerge by removing the bound $0$ from the corresponding definitions of
bound rt-reduction for environments in \secref{lprs}.

\begin{theorem}[decomposition, part one]\
\thslabel{lpxs}
\begin{enumerate}

\item\thlabel{fqu_cpx_trans}
\Caption{confluence of direct subclosure with extended rt-reduction for terms}
If $\FQU{L}{U_1}{K}{T_1}$
and $\CPX{K}{T_1}{T_2}$
then there exists $U_2$\\
such that $\CPX{L}{U_1}{U_2}$
and $\FQU{L}{U_2}{K}{T_2}$.

\item\thlabel{lpx_fqu_trans}
\Caption{confluence of direct subclosure with extended rt-reduction on all entries}
If $\LPX{L_1}{L_2}$
and $\FQU{L_2}{U}{K_2}{T}$
then there exist $K_1$ and $U_0$\\
such that $\CPX{L_1}{U}{U_0}$
and $\FQU{L_1}{U_0}{K_1}{T}$
and $\LPX{K_1}{K_2}$.

\item\thlabel{lpx_cpx_trans}
\Caption{transitivity of extended rt-reduction on all entries and extended rt-reduction for terms}
If $\LPX{K}{L}$
and $\CPX{L}{T_1}{T_2}$
then $\CPXS{K}{T_1}{T_2}$.

\end{enumerate}
\end{theorem}

\begin{proof}
\thref{}{fqu_cpx_trans}
is proved by induction on the first premise
and by cases on the second premise.
with the help of \thref{lsubr}{lsubr_cpx_trans}
in the case of Rule $\ruleref{cpx}{\beta}$.
\thref{}{lpx_fqu_trans} and \thref{}{lpx_cpx_trans}
are proved by induction on the second premise
and by cases on the first premise.
\end{proof}

Notice that $\LPXS{L_1}{L_2}$ has the next alternative axiomatization
because of \thref{lpxs}{lpx_cpx_trans}.
\begin{equation}
\eqnlabel{lpxs_alt}
\vcenter{
\infer[\ruleref{lpxs}{\LAtom}]
{\LPXS{\LAtom}{\LAtom}}
{}
}\sep
\vcenter{
\infer[\ruleref{lpxs}{P}]
{\LPXS{K_1\LPair{y}{V_1}}{K_2\LPair{y}{V_2}}}
{\LPXS{K_1}{K_2}&\CPXS{K_1}{V_1}{V_2}}
}
\end{equation}

\subsection{Inherited Free Variables}
\seclabel{freep}

For a term $T$ which is part of a closure $\Cl{L}{T}$,
we define the finite subset, call it $\FreeP{L}{T}$,
of the inherited free variables of $T$ in $L$ as the union of
the free variables of $T$ and the free variables of the terms in the
entries of $L$ recursively referred by $T$. 
The formal definition of $\FreeP{L}{T}$ is in \tabref{freep}.
In the following we will also use $\f$ and $\g$ for generic finite subsets of variables.

\begin{table}
\seccaption{freep}
{Inherited free variables.}
\tablabel{freep}
\begin{tabular}{c}

$\FreeP{L}{\TSRef{s}} \defeq \SBOT$
\ $\MATOM{\ruleref{freep}{\TSort{}}}$
\sep

$\FreeP{K\LPair{x}{V}}{\TNRef{x}} \defeq
 \FreeP{K}{V} \SOR \SUBSET{x}{}
$
\ $\MATOM{\ruleref{freep}{\TLRef{}}}$
\nl

$\FreeP{\LAtom}{\TNRef{x}} \defeq \SUBSET{x}{}$
\ $\MATOM{\ruleref{freep}{\LAtom}}$
\sep

$\FreeP{K\LPair{y}{V}}{\TNRef{x}} \defeq
 \FreeP{K}{\TNRef{x}}
$
\ $\MATOM{\ruleref{freep}{L}}$
\nl

$\FreeP{L}{\TPair{\Y}{x}{V}T} \defeq
 \FreeP{L}{V} \SOR \FreeP{L\LPair{V}{x}}{T}
 \SDIFF \SUBSET{x}{}
$
\ $\MATOM{\ruleref{freep}{P}}$
\nl

$\FreeP{L}{\TFlat{V}T} \defeq
 \FreeP{L}{V} \SOR \FreeP{L}{T}
$
\ $\MATOM{\ruleref{freep}{F}}$
\nl

Rule $\ruleref{freep}{L}$:
$y \neq x$.
\\

\end{tabular}
\end{table}

With some machinery presented in \appref{exclusion},
we can prove the next equivalent of Rule $\ruleref{freep}{P}$:
\begin{equation}
\eqnlabel{freep}
\FreeP{L}{\TPair{\Y}{x}{V}T} =
\FreeP{L}{V} \SOR \FreeP{L}{T}
\SDIFF \SUBSET{x}{}
\sep\MATOM{\text{$x$ not bound in $L$}}
\end{equation}

\subsection{Syntactic Equivalence on Referred Entries}
\seclabel{req}

As we see in \tabref{cnv} (\secref{nta}),
the validity predicate $\CNV{\A}{L}{T}$
does not depend on the entries of the environment $L$
that are not recursively referred by the term $T$.
Thus, a relation $\REQ{T}{L_1}{L_2}$ equating two environments
just on the basis of their entries recursively referred by a term
emerges naturally.
In order to define it,
we first define $\REQ{\f}{L_1}{L_2}$ in \tabref{req}
for an arbitrary finite subset $\f$ of variables.
Then we use inherited free variables to define
$\REQ{T}{L_1}{L_2}$ as $\REQ{\FreeP{L_1}{T}}{L_1}{L_2}$.

\begin{table}
\seccaption{req}
{Syntactic equivalence for environments (on selected entries).}
\tablabel{req}
\begin{tabular}{c}

\infer[\ruleref{req}{\LAtom}]
{\REQ{\f}{\LAtom}{\LAtom}}
{}
\sep

\infer[\ruleref{req}{B}]
{\REQ{\f}{K_1\LPair{y}{V_1}}{K_2\LPair{y}{V_2}}}
{\REQ{\f}{K_1}{K_2}}
\sep

\infer[\ruleref{req}{P}]
{\REQ{\f\SOR\SUBSET{y}{}}{K_1\LPair{y}{V_1}}{K_2\LPair{y}{V_2}}}
{\REQ{\f}{K_1}{K_2}&\TEQ{V_1}{V_2}}
\nl

Rules $\ruleref{req}{B}$ and $\ruleref{req}{P}$: $y \notin \f$.
\\

\end{tabular}
\end{table}

Indeed, working under the assumption that
every entry of $L$ is valid when $\CNV{\A}{L}{T}$
simplifies the whole meta-theory significantly,
but we aim at showing that this assumption is not necessary.

We term a step from $\Cl{L_1}{T}$ to $\Cl{L_2}{T}$ when $\REQ{T}{L_1}{L_2}$ 
a q-step of our qrst-reduction system.

\begin{theorem}[decomposition, part two]\
\thslabel{req}
\begin{enumerate}

\item\thlabel{req_fqu_trans}
\Caption{confluence of direct subclosure with syntactic equivalence on referred entries}
If $\FQU{L_2}{U}{K_2}{T}$
and $\REQ{U}{L_1}{L_2}$
then there exists $K_1$\\
such that $\FQU{L_1}{U}{K_1}{T}$
and $\REQ{T}{K_1}{K_2}$.

\item\thlabel{cpx_req_conf_sn}
\Caption{extended rt-reduction preserves syntactic equivalence on referred entries}
If $\CPX{L_1}{T_1}{T_2}$
and $\REQ{T_1}{L_1}{L_2}$
then $\REQ{T_2}{L_1}{L_2}$.

\item\thlabel{cpx_req_conf}
\Caption{syntactic equivalence on referred entries preserves extended rt-reduction}
If $\CPX{L_1}{T_1}{T_2}$
and $\REQ{T_1}{L_1}{L_2}$
then $\CPX{L_2}{T_1}{T_2}$.

\item\thlabel{lpx_req_conf}
\Caption{confluence of extended rt-reduction on all entries with syntactic equivalence on referred entries}
If $\LPX{L_0}{L_1}$
and $\REQ{T}{L_0}{L_2}$ 
then there exists $L$
such that $\REQ{T}{L_1}{L}$
and $\LPX{L_2}{L}$.

\end{enumerate}
\end{theorem}

\begin{proof}
\thref{}{req_fqu_trans} and \thref{}{cpx_req_conf}
are proved by induction on the first premise
and by cases on the other premises.
\thref{}{cpx_req_conf_sn}
is proved by induction on the closure $\Cl{L_1}{T_1}$
and by cases on the premises.
\thref{}{lpx_req_conf}
follows with quite an effort from \thref{}{cpx_req_conf}.
\end{proof}

\subsection{QRST-Reduction for Closures}
\seclabel{fpbs}

In this section we introduce our \emph{big trees} extending the ideas of \citeN{SPAc4}.
The nodes of an extended big tree are closures $\Cl{L}{T}$
and the arcs are the qrst-steps denoted by the relation $\FPB{L_1}{T_1}{L_2}{T_2}$
defined by the rules of \tabref{fpb}.
This relation comprises the steps we introduced in
\secref{fqus}, \secref{cpxs}, \secref{lpxs}, \secref{req}
and is reflexive because of Rule $\ruleref{fpb}{lq}$.

\begin{table}
\seccaption{fpbs}
{Extended qrst-reduction for closures (one step).}
\tablabel{fpb}
\begin{tabular}{c}

\infer[\ruleref{fpb}{ex}]
{\FPB{L}{T_1}{L}{T_2}}
{\CPX{L}{T_1}{T_2}}
\sep

\infer[\ruleref{fpb}{lx}]
{\FPB{L_1}{T}{L_2}{T}}
{\LPX{L_1}{L_2}}
\nl

\infer[\ruleref{fpb}{lq}]
{\FPB{L_1}{T}{L_2}{T}}
{\REQ{T}{L_1}{L_2}}
\sep

\infer[\ruleref{fpb}{cs}]
{\FPB{L_1}{T_1}{L_2}{T_2}}
{\FQU{L_1}{T_1}{L_2}{T_2}}
\\

\end{tabular}
\end{table}

The transitive closure $\FPBS{L_1}{T_1}{L_2}{T_2}$,
defined in \tabref{fpbs},
is interesting because of the decomposition \thref{fpbs}{fpbs_inv_star},
\ie in a sequence of qrst-steps
the rt-steps for terms may precede the s-steps,
these may precede the rt-steps for environments and
these may precede the q-steps. 

\begin{theorem}[decomposition, part three]\
\thslabel{fpbs}
\begin{enumerate}

\item\thlabel{fpbs_inv_star}
\Caption{decomposition of qrst-reduction}
If $\FPBS{L_1}{T_1}{L_2}{T_2}$
then there exist $T$, $L$ and $L_0$\\
such that $\CPXS{L_1}{T_1}{T}$
and $\FQUS{L_1}{T}{L}{T_2}$
and $\LPXS{L}{L_0}$
and $\REQ{T_2}{L_0}{L_2}$.

\item\thlabel{fpbs_intro_star}
\Caption{composition of qrst-reduction}
If $\CPXS{L_1}{T_1}{T}$
and $\FQUS{L_1}{T}{L}{T_2}$
and $\LPXS{L}{L_0}$
and $\REQ{T_2}{L_0}{L_2}$
then $\FPBS{L_1}{T_1}{L_2}{T_2}$.

\end{enumerate}
\end{theorem}

\begin{proof}
\thref{}{fpbs_inv_star} (the decomposition theorem) follows from
\thref{lpxs}{fqu_cpx_trans}, \thref{lpxs}{lpx_fqu_trans}, \thref{lpxs}{lpx_cpx_trans},
\thref{req}{req_fqu_trans}, \thref{req}{cpx_req_conf_sn}, \thref{req}{cpx_req_conf}
and \thref{req}{lpx_req_conf}.
\thref{}{fpbs_intro_star}
is the inverse statement of \thref{}{fpbs_inv_star} and is immediate. 
\end{proof}

\section{Strong Normalization of QRST-Reduction and Related Notions}
\seclabel{normalization}

The reader will surely notice that
the generalized t-step of Rule $\ruleref{cpx}{s}$ and
the q-step of Rule $\ruleref{fpb}{lq}$ apply endlessly.
Thus a closure can be strongly normalizing with respect to
qrst-reduction (\secref{fpbs})
only up to an equivalence relation that we define in \secref{feqx}.
With this equivalence set up,
we can define normal forms and strongly normalizing forms in \secref{fsb},
where we prove that a closure $\Cl{L}{T}$ is strongly qrst-normalizing
if $T$ is strongly rt-normalizing in $L$.
So we focus on strongly rt-normalizing terms in \secref{gcr} where we
prove that they form a Tait-style reducibility candidate \cite{Tai75}.
This differs from a Girard-style reducibility candidate \cite{GTL89}
in that it contains untyped terms and
it is not closed by forward r-reduction (Girard's condition CR2).
This simplification gives us more freedom for constructing its elements.
Then we establish our reducibility theorem in \secref{lsubc},
\ie a candidate $C$ contains all terms to which $C$  
can be assigned according to well-established type rules derived from $\LR$.
The types of this assignment, that we interpret as candidates,
are known as arities or norms in the literature.
Normable terms are therefore strongly rt-normalizing and we argue
in the end that they form a decidable superset of the valid terms.
Connecting the previous results, strong normalization of qrst-reduction
from valid closures is immediate.
This is our extension of the so-called big-tree theorem of \citeN{SPAc4}
and its importance is fundamental because it provides for
a very powerful induction principle that is
the basis for proving the most important properties of validity
as we will see in \secref{preservation}.

In the end our system is strongly normalizing because so is $\LR$.
Indeed, we could prove normalization with methods specific to $\LR$
(for example using the sets of computable expressions of \citeN{Tai67}).
Nevertheless, we prefer to present a proof based on reducibility candidates,
which should have the advantage of being reusable 
for some extension of our system evolving beyond $\LR$.

\subsection{Sort Irrelevance for Terms, Environments and Closures}
\seclabel{feqx}

Given that q-steps and generalized t-steps are always possible,
the most sensible definition of a normal form in our qrst-reduction system
is that of a form to which only these steps apply.

Now we observe that,
if the reduction $\FPBS{L_1}{T_1}{L_2}{T_2}$ consists just of these steps,
then $T_1$ and $T_2$ differ only in their sort items.
Moreover $L_1$ and $L_2$ differ only in the sort items
of their entries recursively referred by $T_1$ (or equivalently by $T_2$),
while they can differ at will in the other entries.

So the notion of sort irrelevance emerges naturally
and it is captured by the three equivalence relations
$\TEQX{T_1}{T_2}$ (for terms),
$\REQX{T}{L_1}{L_2}$ (for environments) and
$\FEQX{L_1}{T_1}{L_2}{T_2}$ (for closures)
that we define with the rules of
\tabref{teqx}, \tabref{reqx} and \tabref{feqx} respectively.
In particular, we define $\REQX{T}{L_1}{L_2}$ as $\REQX{\FreeP{L_1}{T}}{L_1}{L_2}$
following the pattern we use for $\REQ{T}{L_1}{L_2}$ in \secref{req}.

\begin{table}
\seccaption{feqx}
{Sort irrelevance for terms.}
\tablabel{teqx}
\begin{tabular}{c}

\infer[\ruleref{teqx}{\TSort{}}]
{\TEQX{\TSRef{s_1}}{\TSRef{s_2}}}
{}
\sep

\infer[\ruleref{teqx}{\TLRef{}}]
{\TEQX{\TNRef{x}}{\TNRef{x}}}
{}
\nl

\infer[\ruleref{teqx}{P}]
{\TEQX{\TPair{\Y}{x}{V_1}T_1}{\TPair{\Y}{x}{V_2}T_2}}
{\TEQX{V_1}{V_2}&\TEQX{T_1}{T_2}}
\sep

\infer[\ruleref{teqx}{F}]
{\TEQX{\TFlat{V_1}T_1}{\TFlat{V_2}T_2}}
{\TEQX{V_1}{V_2}&\TEQX{T_1}{T_2}}
\\

\end{tabular}
\end{table}

\begin{table}
\seccaption{feqx}
{Sort irrelevance for environments (on selected entries).}
\tablabel{reqx}
\begin{tabular}{c}

\infer[\ruleref{reqx}{\LAtom}]
{\REQX{\f}{\LAtom}{\LAtom}}
{}
\sep

\infer[\ruleref{reqx}{B}]
{\REQX{\f}{K_1\LPair{y}{V_1}}{K_2\LPair{y}{V_2}}}
{\REQX{\f}{K_1}{K_2}}
\sep

\infer[\ruleref{reqx}{P}]
{\REQX{\f\SOR\SUBSET{y}{}}{K_1\LPair{y}{V_1}}{K_2\LPair{y}{V_2}}}
{\REQX{\f}{K_1}{K_2}&\TEQX{V_1}{V_2}}
\nl

Rules $\ruleref{reqx}{B}$ and $\ruleref{reqx}{P}$: $y \notin \f$.
\\

\end{tabular}
\end{table}

\begin{table}
\seccaption{feqx}
{Sort irrelevance for closures.}
\tablabel{feqx}
\begin{tabular}{c}

\infer[\ruleref{feqx}{I}]
{\FEQX{L_1}{T_1}{L_2}{T_2}}
{\REQX{T_1}{L_1}{L_2}&\TEQX{T_1}{T_2}}
\\

\end{tabular}
\end{table}

We want to stress that the small-step reductions we are using in our system
make the converse hold as well.
In particular, if $\FEQX{L_1}{T_1}{L_2}{T_2}$, then
only q-steps and generalized t-steps are possible in the reduction
$\FPB{L_1}{T_1}{L_2}{T_2}$.
So one-step cycles are impossible up to sort irrelevance.

\subsection{Normal Forms and Strongly Normalizing Forms for QRST-Reduction}
\seclabel{fsb}

The key observation we made in \secref{feqx}
that one-step cycles are impossible up to sort irrelevance
gives us the right to consider
a term $T$ rt-normal in the environment $L$
when it satisfies the condition $\CNX{L}{T}$
defined by Rule $\ruleref{csx}{N}$ of \tabref{csx}.
The condition $\CSX{L}{T}$ satisfied by
a strongly rt-normalizing term $T$ in $L$ follows immediately
and is defined by Rule $\ruleref{csx}{S}$ of \tabref{csx}.

Notice that the non-recursive introduction rule of this relation is:
if $\CNX{L}{T}$ then $\CSX{L}{T}$.

\begin{table}
\seccaption{fsb}
{Normal and strongly normalizing terms for extended rt-reduction.}
\tablabel{csx}
\begin{tabular}{c}

\infer[\ruleref{csx}{N}]
{\CNX{L}{T_1}}
{\MALL{T_2}
 \MATOM{\CPX{L}{T_1}{T_2}} \MIMP
 \MATOM{\TEQX{T_1}{T_2}}
}
\sep

\infer[\ruleref{csx}{S}]
{\CSX{L}{T_1}}
{\MALL{T_2}
 \MATOM{\CPX{L}{T_1}{T_2}} \MIMP
 \MATOM{\TNEQX{T_1}{T_2}} \MIMP
 \MATOM{\CSX{L}{T_2}}
}
\\

\end{tabular}
\end{table}

Following the same pattern, we introduce
a strongly qrt-normalizing environment $L$ with respect to a term $T$
with the predicate $\RSX{T}{L}$ defined in \tabref{rsx}
and moreover a strongly qrst-normalizing closure $\Cl{L}{T}$
with the predicate $\FSB{L}{T}$ defined in \tabref{fsb}.

\begin{table}
\seccaption{fsb}
{Strongly normalizing environments for extended qrt-reduction (on referred entries).}
\tablabel{rsx}
\begin{tabular}{c}

\infer[\ruleref{rsx}{S}]
{\RSX{T}{L_1}}
{\MALL{L_2}
 \MATOM{\LPX{L_1}{L_2}} \MIMP
 \MATOM{\RNEQX{T}{L_1}{L_2}} \MIMP
 \MATOM{\RSX{T}{L_2}}
}
\\

\end{tabular}
\end{table}

\begin{table}
\seccaption{fsb}
{Strongly normalizing closures for extended qrst-reduction.}
\tablabel{fsb}
\begin{tabular}{c}

\infer[\ruleref{fsb}{S}]
{\FSB{L_1}{T_1}}
{\MALL{L_2, T_2}
 \MATOM{\FPB{L_1}{T_1}{L_2}{T_2}} \MIMP
 \MATOM{\FNEQX{L_1}{T_1}{L_2}{T_2}} \MIMP
 \MATOM{\FSB{L_2}{T_2}}
}
\\

\end{tabular}
\end{table}

Unfortunately, the introduction Rule $\ruleref{rsx}{P}$ for $\RSX{T}{L}$
differs from the other $P$-rules in that
the entry $\LPair{x}{V}$ disappears in the conclusion
of its right inverse $\ruleref{rsx}{P\dx}$. Thus $L$ must not bind $x$ in it.
\begin{equation}
\eqnlabel{rsxP}
\vcenter{
\infer[\ruleref{rsx}{P}]
{\RSX{\TPair{\Y}{x}{V}T}{L}}
{\RSX{V}{L}&\RSX{T}{L\LPair{x}{V}}}
}\sep
\vcenter{
\infer[\ruleref{rsx}{P\sn}]
{\RSX{V}{L}}
{\RSX{\TPair{\Y}{x}{V}T}{L}}
}\sep
\vcenter{
\infer[\ruleref{rsx}{P\dx}]
{\RSX{T}{L}}
{\RSX{\TPair{\Y}{x}{V}T}{L}}
}
\end{equation}

This means that when $\RSX{T_1}{L}$ and $\CPX{L}{T_1}{T_2}$ appear
as premises in a statement proved by induction, for instance:
if $\CPX{L}{T_1}{T_2}$ then $\RSX{T_1}{L}$ implies $\RSX{T_2}{L}$,
the premise $\RSX{T_1}{L}$ must be generalized as $\RSX{T_1}{K}$
and a compatibility relation $\JSX{L}{K}$ between the two
environments must be set up.
The result is the next \thref{fsb}{rsx_cpx_trans_jsx}.
This reflexive and transitive relation,
that we define with the rules of \tabref{jsx},
accommodates the proof when $T_1 = \TNRef{x}$ or $T_1 = \TPair{\Y}{x}{V}T$.

\begin{table}
\seccaption{fsb}
{Compatibility relation for strongly normalizing environments.}
\tablabel{jsx}
\begin{tabular}{c}

\infer[\ruleref{jsx}{\LAtom}]
{\JSX{\LAtom}{\LAtom}}
{}
\sep

\infer[\ruleref{jsx}{B}]
{\JSX{K_1\LPair{y}{V}}{K_2\LPair{y}{V}}}
{\JSX{K_1}{K_2}}
\sep

\infer[\ruleref{jsx}{P}]
{\JSX{K_1\LPair{y}{V}}{K_2}}
{\JSX{K_1}{K_2}&\RSX{K_2}{V}}
\nl

Rule $\ruleref{jsx}{P}$: $y$ not bound in $K_2$.\\

\end{tabular}
\end{table}

\begin{theorem}[strongly normalizing forms, part one]\
\thslabel{fsb}
\begin{enumerate}

\item\thlabel{rsx_cpx_trans_jsx}
\Caption{rt-reduction for terms preserves normalization with respect to rt-reduction on referred entries}
If $\CPX{L}{T_1}{T_2}$ and $\JSX{L}{K}$ 
and $\RSX{T_1}{K}$ then $\RSX{T_2}{K}$.

\item\thlabel{rsx_lref_pair_lpxs}
\Caption{normalization with respect to rt-reduction of an entry referred by a variable occurrence}
If $\CSX{K_1}{V}$
and $\RSX{V}{K_2}$
and $\LPXS{K_1}{K_2}$
then $\RSX{\TNRef{x}}{K_2\LPair{x}{V}}$.

\item\thlabel{csx_rsx}
\Caption{normalization for terms implies normalization on referred entries} 
If $\CSX{L}{T}$ then $\RSX{T}{L}$.

\item\thlabel{csx_fsb_fpbs}
\Caption{normalization with respect to rt-reduction implies normalization with respect to qrst-reduction}
If $\CSX{L_1}{T_1}$
and $\FPBS{L_1}{T_1}{L_2}{T_2}$ 
then $\FSB{L_2}{T_2}$.

\end{enumerate}
\end{theorem}

\begin{proof}
\thref{}{rsx_cpx_trans_jsx}
is proved by induction on the first premise
and by cases on the other premises
with some invocations of Rule $\ruleref{rsx}{P\dx}$,
\thref{req}{lpx_req_conf} and \thref{lpxs}{lpx_cpx_trans}.
\thref{}{rsx_lref_pair_lpxs}
is proved by induction on the first premise,
then by induction on the second premise
and by cases on the third premise
by means of \thref{}{rsx_cpx_trans_jsx} and \thref{lpxs}{lpx_cpx_trans}.
\thref{}{csx_rsx}
is proved by induction on the closure $\Cl{L}{T}$
with the help of \thref{}{rsx_lref_pair_lpxs}
when $L = K\LPair{x}{V}$ and $T = \TNRef{x}$.
\thref{}{csx_fsb_fpbs} is proved by induction on the first premise
with the help of \thref{}{csx_rsx} and \thref{fpbs}{fpbs_inv_star}.
\end{proof}

Notice in the end that without loss of generality we can replace
$\CPX{L}{T_1}{T_2}$ with $\CPXS{L}{T_1}{T_2}$ in
Rule $\ruleref{csx}{N}$ and Rule $\ruleref{csx}{S}$.
We can also replace $\LPX{L_1}{L_2}$ with $\LPXS{L_1}{L_2}$
in Rule $\ruleref{rsx}{S}$
and moreover we can replace
$\FPB{L_1}{T_1}{L_2}{T_2}$ with $\FPBS{L_1}{T_1}{L_2}{T_2}$
in Rule $\ruleref{fsb}{S}$.

\subsection{Reducibility Candidates, Neutral Terms and Inner RT-Reduction}
\seclabel{gcr}

The most important property of strongly rt-normalizing terms is that
$\CSX{}{} \defeq \SUBSET{\Cl{L}{T}}{\CSX{L}{T}}$ is a reducibility candidate
(\thref{gcr}{csx_gcr}),
\ie a subset $\C$ of $\CSX{}{}$ verifying the six saturation conditions
expressed by the rules of \tabref{gcr}.
In this definition we use some vector notation:
$\V{V}$ is a list of terms,
$\TApplV{\V{V}}T$ is the n-ary application
and $\CSXV{L}{\V{V}}$ means $\CSX{L}{V}$ for every $V \in \V{V}$.

\begin{table}
\seccaption{gcr}
{Reducibility candidate.}
\tablabel{gcr}
\begin{tabular}{c}

\infer[\ruleref{gcr}{J}]
{\CSX{L}{T}}
{\Cl{L}{T} \in \C}
\sep

\infer[\ruleref{gcr}{\TLRef{}}]
{\Cl{L}{\TApplV{\V{V}}\TNRef{x}} \in \C}
{\Drops{L}{K\LPair{x}{V}}&\Cl{L}{\TApplV{\V{V}}V} \in \C}
\nl

\infer[\ruleref{gcr}{N}]
{\Cl{L}{\TApplV{\V{V}}T} \in \C}
{\CSXV{L}{\V{V}}&\Simple{T}&\CNX{L}{T}}
\sep

\infer[\ruleref{gcr}{\beta}]
{\Cl{L}{\TApplV{\V{V}}\TAppl{V}\TAbst{\Y}{x}{W}T} \in \C}
{\Cl{L}{\TApplV{\V{V}}\TAbbr{x}{\TCast{W}V}T} \in \C}
\nl

\infer[\ruleref{gcr}{\TAbbr{}{}}]
{\Cl{L}{\TApplV{\V{V}}\TAbbr{x}{V}T} \in \C}
{\CSX{L}{V} & \Cl{L\LAbbr{x}{V}}{\TApplV{\V{V}}T} \in \C}
\sep

\infer[\ruleref{gcr}{\TCast{}}]
{\Cl{L}{\TApplV{\V{V}}\TCast{U}T} \in \C}
{\Cl{L}{\TApplV{\V{V}}U} \in \C & \Cl{L}{\TApplV{\V{V}}T} \in \C}
\\

\end{tabular}
\end{table}

Rule $\ruleref{gcr}{J}$ is Girard's condition CR1,  
Rule $\ruleref{gcr}{N}$ is Tait's condition iii or Girard's condition CR4, 
Rule $\ruleref{gcr}{\beta}$ generalizes Tait's condition ii.
The neutral term $T$ in Rule $\ruleref{gcr}{N}$ does not make a redex
with a preceding application,
so it is neither an abstraction (because of $\beta$-reduction)
nor a definition (because of $\theta$-reduction).
The formal definition of $\Simple{T}$ is given by the rules of \tabref{simple}.

When we say in Rule $\ruleref{gcr}{\TLRef{}}$ that
the environment $L$ extends the environment $K$,
we mean that all entries $\LPair{x}{V}$ of $K$ are in $L$ in the same order.
Moreover, the entries of $L$ that are not in $K$
do not depend on the variables bound in $K$.
We stress that the explicit substitution occurring in 
the second premise of Rule $\ruleref{gcr}{\delta}$
prevents us from assuming that $K\LPair{x}{V}$ is simply an initial segment of $L$.

Compound reducibility candidates are built through well-established constructions.
In this case we need just the functional construction $\CFun{\C_1}{\C_2}$
defined by the rule of \tabref{cfun},
whose relevant property is \thref{gcr}{acr_gcr} stating that
if $\C_1$ and $C_2$ are candidates, so is $\CFun{\C_1}{\C_2}$. 
Notice that $L$ must extend $K$ to prove the case of
Rule $\ruleref{gcr}{\delta}$,
in which $L$ and $L\LAbbr{x}{V}$ have different lengths.

\begin{table}
\seccaption{gcr}
{Functional construction for subsets of closures.}
\tablabel{cfun}
\begin{tabular}{c}

\infer[\ruleref{cfun}{I}]
{\Cl{K}{T} \in \CFun{\C_1}{\C_2}}
{\MALL{L,V}
 \MATOM{\Drops{L}{K}} \MIMP
 \MATOM{\Cl{L}{V} \in \C_1} \MIMP
 \MATOM{\Cl{L}{\TAppl{V}T} \in \C_2}
}
\nl

Rule $\ruleref{cfun}{I}$:
$T$ does not depend on the variables bound in $L$ but not in $K$.
\\

\end{tabular}
\end{table}

When we prove that $\CSX{}{}$ is a reducibility candidate,
for example in the case of Rule $\ruleref{gcr}{\beta}$,
we need to separate the reductions involving the top redex
(the outer reductions) from the reductions preserving the top redex
(the inner reductions).
Here it is enough to say that
$\CPXS{L}{T_1}{T_2}$ is inner when $T_1$ and $T_2$ have the same top construction
in the sense of $\TEQO{T_1}{T_2}$ defined by the rules of \tabref{teqo}.
For the sake of elegance $\TEQX{T_1}{T_2}$ implies $\TEQO{T_1}{T_2}$
and we show its usage in \thref{gcr}{cpxs_fwd_beta}.

\begin{table}
\seccaption{gcr}
{Equivalence for inner rt-reduction.}
\tablabel{teqo}
\begin{tabular}{c}

\infer[\ruleref{teqo}{\TSort{}}]
{\TEQO{\TSRef{s_1}}{\TSRef{s_2}}}
{}
\sep

\infer[\ruleref{teqo}{\TLRef{}}]
{\TEQO{\TNRef{x}}{\TNRef{x}}}
{}
\nl

\infer[\ruleref{teqo}{P}]
{\TEQO{\TPair{\Y}{x}{V_1}T_1}{\TPair{\Y}{x}{V_2}T_2}}
{}
\sep

\infer[\ruleref{teqo}{F}]
{\TEQO{\TFlat{V_1}T_1}{\TFlat{V_2}T_2}}
{}
\\

\end{tabular}
\end{table}

\begin{theorem}[reducibility candidates]\
\thslabel{gcr}
\begin{enumerate}

\item\thlabel{cpxs_fwd_beta}
\Caption{the extended rt-reducts of the $\beta$-redex}
If $\CPXS{L}{\TAppl{V}\TAbst{\Y}{x}{W}T_1}{T_2}$
then either $\TEQO{\TAppl{V}\TAbst{\Y}{x}{W}T_1}{T_2}$
or $\CPXS{L}{\TAbbr{x}{\TCast{W}V}T_1}{T_2}$.

\item\thlabel{csx_gcr}
\Caption{strongly normalizing terms with respect to extended rt-reduction form a reducibility candidate}
$\CSX{}{}$ is a reducibility candidate.

\item\thlabel{acr_gcr}
\Caption{the functional reducibility candidate}
If $\C_1$ and $\C_2$ are reducibility candidate, so is $\CFun{\C_1}{\C_2}$.

\end{enumerate}
\end{theorem}

\begin{proof}
\thref{}{cpxs_fwd_beta}
is proved by cases on $T_2$.
We need \thref{lsubr}{lsubr_cpx_trans} and \thref{lpxs}{lpx_cpx_trans}
when $T_2$ comes from a $\beta$-contraction.
\thref{}{csx_gcr} (the normalization theorem)
is proved by \thref{req}{cpx_req_conf} and by
lemmas like \thref{}{cpxs_fwd_beta}.
\thref{}{acr_gcr} is immediate.
\end{proof}

\subsection{Arity Assignment and Related Refinements}
\seclabel{lsubc}

Arities, also known as norms, are simple types representing
the abstract syntax of our reducibility candidates,
which are built from $\C \defeq \CSX{}{}$ and $\C \defeq \CFun{\C_1}{\C_2}$
of \secref{gcr}.
We introduce them in \tabref{acr} with their interpretation and
we assign them to terms according to the rules of \tabref{aaa},
that are an adaptation to our system of the well-established type rules of $\LR$.

\begin{table}
\seccaption{lsubc}
{Arities and their interpretation as reducibility candidates.}
\tablabel{acr}
\begin{tabular}{lr@{\;}l}
Arity:& 
$A,B \GDEF$&
$
\AAtom \GOR
\APair{B}{A}
$\\ 
\end{tabular}\nl

\begin{tabular}{c}
$\ACR{\AAtom} \defeq \CSX{}{}$\sep
$\ACR{\APair{B}{A}} \defeq \CFun{\ACR{B}}{\ACR{A}}$\\
\end{tabular}
\end{table}

\begin{table}
\seccaption{lsubc}
{Arity assignment.}
\tablabel{aaa}
\begin{tabular}{c}

\infer[\ruleref{aaa}{\TSort{}}]
{\AAA{L}{\TSRef{s}}{\AAtom}}
{}
\sep

\infer[\ruleref{aaa}{\TLRef{}}]
{\AAA{K\LPair{x}{V}}{\TNRef{x}}{A}}
{\AAA{K}{V}{A}}
\sep

\infer[\ruleref{aaa}{L}]
{\AAA{K\LPair{y}{V}}{\TNRef{x}}{A}}
{\AAA{K}{\TNRef{x}}{A}}
\nl

\infer[\ruleref{aaa}{\TAbst{}{}{}}]
{\AAA{L}{\TAbst{\Y}{x}{W}T}{\APair{B}{A}}}
{\AAA{L}{W}{B}&\AAA{L\LAbst{x}{W}}{T}{A}}
\sep

\infer[\ruleref{aaa}{\TAbbr{}{}}]
{\AAA{L}{\TAbbr{x}{V}T}{A}}
{\AAA{L}{V}{B}&\AAA{L\LAbbr{x}{V}}{T}{A}}
\nl

\infer[\ruleref{aaa}{\TAppl{}}]
{\AAA{L}{\TAppl{V}T}{A}}
{\AAA{L}{V}{B}&\AAA{L}{T}{\APair{B}{A}}}
\sep

\infer[\ruleref{aaa}{\TCast{}}]
{\AAA{L}{\TCast{U}T}{A}}
{\AAA{L}{U}{A}&\AAA{L}{T}{A}}
\nl

Rule $\ruleref{aaa}{L}$: $y \neq x$.\\

\end{tabular}
\end{table}

Arities clarify that our valid terms can have a simple type
with one base type $\AAtom$ and that our strong normalization
amounts to strong normalization in $\LR$.
In this respect the type assignment $\AAA{L}{T}{A}$ has two interpretations:
either $A$ is the simple type of the object $T$,
or $A$ is the simple type associated to the type $T$. 
\thref{lsuba}{} shows that this type assignment has important expected properties
like uniqueness, preservation by extended rt-reduction, decidability. 

The proof of preservation requires the refinement $\LSubA{L_1}{L_2}$
defined by the rules of \tabref{lsuba},
whose main properties are
\thref{lsuba}{lsuba_aaa_conf} and \thref{lsuba}{lsuba_aaa_trans}
that make $\LSubA{L_1}{L_2}$ transitive.

\begin{table}
\seccaption{lsubc}
{Refinement for preservation of arity.}
\tablabel{lsuba}
\begin{tabular}{c}

\infer[\ruleref{lsuba}{\LAtom}]
{\LSubA{\LAtom}{\LAtom}}
{}
\sep

\infer[\ruleref{lsuba}{B}]
{\LSubA{K_1\LPair{y}{V}}{K_2\LPair{y}{V}}}
{\LSubA{K_1}{K_2}}
\sep

\infer[\ruleref{lsuba}{\beta}]
{\LSubA{K_1\LAbbr{y}{\TCast{W}V}}{K_2\LAbst{y}{W}}}
{\LSubA{K_1}{K_2}&\AAA{K_1}{\TCast{W}V}{B}&\AAA{K_2}{W}{B}}
\\

\end{tabular}
\end{table}

\begin{theorem}[arities]\
\thslabel{lsuba}
\begin{enumerate}

\item\thlabel{aaa_mono}
\Caption{the arity of a term is unique}
If $\AAA{L}{T}{A_1}$
and $\AAA{L}{T}{A_2}$
then $A_1 = A_2$.

\item\thlabel{aaa_dec}
\Caption{arity inference is decidable}
We can decide $\MEX{A} \AAA{L}{T}{A}$.

\item\thlabel{lsuba_aaa_conf}
\Caption{weakening of arity assignment through its refinement}
If $\LSubA{L}{K}$
and $\AAA{L}{T}{A}$
then $\AAA{K}{T}{A}$.

\item\thlabel{lsuba_aaa_trans}
\Caption{strengthening of arity assignment through its refinement}
If $\LSubA{K}{L}$
and $\AAA{L}{T}{A}$
then $\AAA{K}{T}{A}$.

\item\thlabel{cpx_aaa_conf_lpx}
\Caption{arity is preserved by extended rt-reduction}
If $\AAA{L_1}{T_1}{A}$
and $\CPX{L_1}{T_1}{T_2}$
and $\LPX{L_1}{L_2}$
then $\AAA{L_2}{T_2}{A}$.

\end{enumerate}
\end{theorem}

\begin{proof}
\thref{}{aaa_mono}, \thref{}{lsuba_aaa_conf} and \thref{}{lsuba_aaa_trans}
are proved by induction on the second premise
and by cases on the first premise.
Notice that \thref{}{lsuba_aaa_conf} and \thref{}{lsuba_aaa_trans}
depend on \thref{}{aaa_mono}.
\thref{}{aaa_dec}
is proved by induction on the closure $\Cl{L}{T}$
with some invocations of \thref{}{aaa_mono}.
\thref{}{cpx_aaa_conf_lpx}
is proved by induction on the first premise
and by cases on the other premises.
Notice that \thref{}{lsuba_aaa_trans}
is needed in the case of Rule $\ruleref{cpx}{\beta}$.
\end{proof}

Arities are the main ingredient of our reducibility theorem,
which states that $\AAA{L}{T}{A}$ implies $\Cl{L}{T} \in \ACR{A}$,
but we must generalize it as \thref{lsubc}{acr_aaa_lsubc_lifts}
in order to obtain suitable inductive hypotheses.
The refinement $\LSubC{L_1}{L_2}$ appearing in that statement
is defined by the rules of \tabref{lsubc}
and expresses a simultaneous substitution like the one we find 
in the reducibility theorem for $\SF$,
which is stated using the so-called parametric reducibility of \citeN{GTL89}.


\begin{table}
\seccaption{lsubc}
{Refinement for reducibility.}
\tablabel{lsubc}
\begin{tabular}{c}

\infer[\ruleref{lsubc}{\LAtom}]
{\LSubC{\LAtom}{\LAtom}}
{}
\sep

\infer[\ruleref{lsubc}{B}]
{\LSubC{K_1\LPair{y}{V}}{K_2\LPair{y}{V}}}
{\LSubC{K_1}{K_2}}
\nl

\infer[\ruleref{lsubc}{\beta}]
{\LSubC{K_1\LAbbr{y}{\TCast{W}V}}{K_2\LAbst{y}{W}}}
{\LSubC{K_1}{K_2}&\Cl{K_1}{W}\in\ACR{B}&\Cl{K_1}{V}\in\ACR{B}&\AAA{K_2}{W}{B}}
\\

\end{tabular}
\end{table}

Finally, the picture of our big-tree \thref{lsubc}{cnv_fwd_fsb}
is completed by \thref{lsubc}{cnv_fwd_aaa} stating that
$\CNV{\A}{L}{T}$ implies $\AAA{L}{T}{A}$ for some $A$,
\ie every valid term is normable in its environment.

We stress that the converse does not hold
since there are normable terms which are not valid.

\begin{theorem}[strongly normalizing forms, part two]\
\thslabel{lsubc}
\begin{enumerate}

\item\thlabel{cnv_fwd_aaa}
\Caption{valid terms have an arity}
If $\CNV{\A}{L}{T}$
then there exists $A$
such that $\AAA{L}{T}{A}$.

\item\thlabel{acr_aaa_lsubc_lifts}
\Caption{terms with an arity belong to the corresponding reducibility candidate}
If $\LSubC{L_1}{L_2}$
and $\Drops{L_2}{K_2}$
and $\AAA{K_2}{T}{A}$
then $\Cl{L_1}{T}\in\ACR{A}$.

\item\thlabel{cnv_fwd_fsb}
\Caption{valid terms are strongly normalizing with respect to qrst-reduction}
If $\CNV{\A}{L}{T}$
then $\FSB{L}{T}$.

\end{enumerate}
\end{theorem}

\begin{proof}
\thref{}{cnv_fwd_aaa}
is proved by induction on the premise
with invocations of \thref{lsuba}{aaa_mono} and \thref{lsuba}{cpx_aaa_conf_lpx}
in the cases of Rule $\ruleref{aaa}{\TAppl{}}$
and Rule $\ruleref{aaa}{\TCast{}}$.
\thref{}{acr_aaa_lsubc_lifts} (the reducibility theorem)
is proved by induction on the closure $\Cl{K_2}{T}$
and by cases on the premises
with an invocation of \thref{lsuba}{aaa_mono}.
\thref{}{cnv_fwd_fsb} (the very-big-tree theorem)
follows from \thref{fsb}{csx_fsb_fpbs}, \thref{gcr}{csx_gcr},
\thref{}{acr_aaa_lsubc_lifts} and \thref{}{cnv_fwd_aaa}.
Notice that $\Cl{L}{T}\in\ACR{A}$ implies $\CSX{L}{T}$ by Rule $\ruleref{gcr}{J}$,
\end{proof}

\section{Preservation and Confluence for Bound RT-Reduction}
\seclabel{preservation}

In this section we present two mutually dependent main results on bound rt-reduction
(\secref{lsubv}) and their general consequences on the calculus.
Remarkably, among these we find the decidability of parametric validity
(\secref{cpes}) that stands on the existence of suitable normal forms
(\secref{cpmuwe}).

Other consequences include the fact that the type judgment
we defined in \secref{nta} has many desired properties (\secref{types}).
An axiomatization of this judgment in the general case,
when the applicability domain is not fixed in advance,
requires to introduce iterated typing (\secref{ntas}).

\subsection{Bound RT-Reduction of Valid Terms and Related Notions}
\seclabel{lsubv}

In this section we discuss the fact that bound rt-reduction
(\secref{cpms}) preserves parametric validity (\secref{nta}),
\ie if $\CNV{\A}{L}{T_1}$ and $\CPM{n}{L}{T_1}{T_2}$ then $\CNV{\A}{L}{T_2}$.
A difficulty arises because this property mutually depends on
another interesting invariant: the confluence of rt-reduction,
\ie if $\CNV{\A}{L}{T_0}$ then $\CPMS{n_1}{L}{T_0}{T_1}$ and
$\CPMS{n_2}{L}{T_0}{T_2}$ imply $\CPMS{n_2-n_1}{L}{T_1}{T}$ and
$\CPMS{n_1-n_2}{L}{T_2}{T}$ for some term $T$. Thus the two properties
must be proved together. Moreover the proof needs the well-founded induction
on the qrst-reducts of $\Cl{L}{T_0}$ provided by the very-big-tree
\thref{lsubc}{cnv_fwd_fsb} (\secref{lsubc}).
Notice that in our setting $0-n = 0$, thus $n_1-n_2 = 0$ or $n_2-n_1 = 0$.
Furthermore, notice that the premise $\CNV{\A}{L}{T_0}$
is essential in this confluence theorem since the critical pair
$(\epsilon, e)$ on $\Cl{L}{\TCast{U_0}T_0}$ may not be confluent
without $\CNV{\A}{L}{\TCast{U_0}T_0}$.
We remark that in the special case of r-reduction
(\secref{lprs}), where $n_1 = n_2 = 0$, the premise is not
needed because the critical pair does not appear.
Another difficulty arises since bound rt-reduction has the kite property
rather than the diamond property, hence we cannot reach its confluence
through a so-called strip lemma.

Intuitively, if we are proving the confluence
of $\CPMSsn{n_1}{m_1}{L}{T_0}{T_1}{U_1}$
and $\CPMSsn{n_2}{m_2}{L}{T_0}{T_2}{U_2}$,
we can apply the kite lemma to the closure $\Cl{L}{T_0}$,
but the inductive hypotheses
apply to $\Cl{L}{T_1}$ only if $\TNEQX{T_0}{T_1}$ and 
apply to $\Cl{L}{T_2}$ only if $\TNEQX{T_0}{T_2}$.
In either case we can easily proceed, but
a difficulty arises when $\TEQX{T_0}{T_1}$ and $\TEQX{T_0}{T_2}$.
In that event we are lucky because the reduction $\CPM{n}{L}{T_1}{T_2}$
with the restriction $\TEQX{T_1}{T_2}$ has the diamond property 
(\thref{lsubv}{cnv_cpm_teqx_conf_lpr_aux}),
and not just the kite property.
So the confluence of its sequences easily comes from a strip lemma.
In addition, by \thref{lsubv}{cnv_cpm_teqx_cpm_trans_sub}
we can move these restricted reductions at the end of reduction sequences. 

We stress that the mentioned diamond property holds because
just generalized t-steps of Rule $\ruleref{cpx}{s}$
are possible in the aforesaid restricted reduction,
thus the critical pair $(\epsilon, e)$ does not appear.

As a side remark, notice that the confluence theorem
stands on the next arithmetical property:
\begin{equation}
\eqnlabel{arith_l4}
(n_2+m_2)-(n_1+m_1) = (n_2-n_1)-m_1+(m_2-(n_1-n_2)-(m_1-(n_2-n_1)))
\end{equation}

Following \citeN{Dln80}, we introduce in \tabref{sub}
some definitions for the predicates of interest.
They are $\PropD{}{}{}$ (diamond), $\PropS{}{}{}$ (swap),
$\PropK{}{}{}$ (kite), $\PropC{}{}{}$ (confluence)
and $\PropP{}{}{}$ (preservation).

\begin{table}
\seccaption{lsubv}
{Predicates for preservation and confluence of bound rt-reduction.}
\tablabel{sub}
\begin{tabular}{lll}

$\PropD{\A}{L_0}{T_0}$&is&
if $\CNV{\A}{L_0}{T_0}$
and $\LPR{L_0}{L_1}$
and $\LPR{L_0}{L_2}$\\&&
and $\CPM{n_1}{L_0}{T_0}{T_1}$ and $\TEQX{T_0}{T_1}$\\&&
and $\CPM{n_2}{L_0}{T_0}{T_2}$ and $\TEQX{T_0}{T_2}$ and $\TNEQX{T_1}{T_2}$
then there exists $T$\\&&
such that $\CPM{n_2-n_1}{L_1}{T_1}{T}$ and $\TEQX{T_1}{T}$
and $\CPM{n_1-n_2}{L_2}{T_2}{T}$ and $\TEQX{T_2}{T}$\nl

$\PropS{\A}{L}{T_1}$&is&
if $\CNV{\A}{L}{T_1}$
and $\CPM{n_1}{L}{T_1}{T_0}$
and $\TEQX{T_1}{T_0}$
and $\CPM{n_2}{L}{T_0}{T_2}$\\&&
then there exists $T$
such that $\CPM{n_2}{L}{T_1}{T}$
and $\CPM{n_1}{L}{T}{T_2}$
and $\TEQX{T}{T_2}$\nl

$\PropK{\A}{L_0}{T_0}$&is&
if $\CNV{\A}{L_0}{T_0}$
and $\LPR{L_0}{L_1}$
and $\LPR{L_0}{L_2}$\\&&
and $\CPM{n_1}{L_0}{T_0}{T_1}$
and $\CPM{n_2}{L_0}{T_0}{T_2}$ 
then there exists $T$\\&&
such that $\CPMS{n_2-n_1}{L_1}{T_1}{T}$
and $\CPMS{n_1-n_2}{L_2}{T_2}{T}$\nl

$\PropC{\A}{L_0}{T_0}$&is&
if $\CNV{\A}{L_0}{T_0}$
and $\LPR{L_0}{L_1}$
and $\LPR{L_0}{L_2}$\\&&
and $\CPMS{n_1}{L_0}{T_0}{T_1}$
and $\CPMS{n_2}{L_0}{T_0}{T_2}$ 
then there exists $T$\\&&
such that $\CPMS{n_2-n_1}{L_1}{T_1}{T}$
and $\CPMS{n_1-n_2}{L_2}{T_2}{T}$\nl

$\PropP{\A}{L_1}{T_1}$&is&
if $\CNV{\A}{L_1}{T_1}$
and $\LPR{L_1}{L_2}$
and $\CPM{n}{L_1}{T_1}{T_2}$
then $\CNV{\A}{L_2}{T_2}$\nl

\multicolumn{3}{c}{Variables on the right
not occurring on the left are universally quantified.}\\

\end{tabular}
\end{table}

Moreover, we introduce the handy relation
$\FPBG{T_1}{L_1}{T_2}{L_2}$
defined by the rule of \tabref{fpbg}.
These definitions simplify the statements of the main \thref{lsubv}{cnv_preserve}
and of its prerequisites.

It might be the case that
the premise $\CNV{\A}{L_0}{T_0}$ of $\PropD{\A}{L_0}{T_0}$ and
the premise $\CNV{\A}{L}{T_1}$ of $\PropS{\A}{L}{T_1}$
could be avoided. 
Anyway, the proofs we developed so far use these premises indeed.

\begin{table}
\seccaption{lsubv}
{Extended qrst-reduction for closures (proper sequence of steps).}
\tablabel{fpbg}
\begin{tabular}{c}

\infer[\ruleref{fpbg}{I}]
{\FPBG{L_1}{T_1}{L_4}{T_4}}
{\FPBS{L_1}{T_1}{L_2}{T_2}&
 \FPB{L_2}{T_2}{L_3}{T_3}&
 \FPBS{L_3}{T_3}{L_4}{T_4}&\FNEQX{L_2}{T_2}{L_3}{T_3}}
\\

\end{tabular}
\end{table}

In the proof of the preservation \thref{lsubv}{cnv_cpm_trans_lpr_aux},
\ie $\MALL{L,T} \PropP{\A}{L}{T}$ under hypotheses, a refinement
$\LSubV{\A}{L_1}{L_2}$ is needed to handle the case of the $\beta$-rule
as it happens for \thref{lprs}{cpr_conf_lpr} (\secref{lprs})
and for \thref{lsuba}{cpx_aaa_conf_lpx} (\secref{lsubc}).
This refinement is defined by the rules of \tabref{lsubv}
and, notably, is transitive because of \thref{lsubv}{lsubv_cnv_trans},
which is its most important property.

\begin{table}
\seccaption{lsubv}
{Refinement for the preservation of validity.}
\tablabel{lsubv}
\begin{tabular}{c}

\infer[\ruleref{lsubv}{\LAtom}]
{\LSubV{\A}{\LAtom}{\LAtom}}
{}
\sep

\infer[\ruleref{lsubv}{B}]
{\LSubV{\A}{K_1\LPair{x}{V}}{K_2\LPair{x}{V}}}
{\LSubV{\A}{K_1}{K_2}}
\sep

\infer[\ruleref{lsubv}{\beta}]
{\LSubV{\A}{K_1\LAbbr{x}{\TCast{W}V}}{K_2\LAbst{x}{W}}}
{\LSubV{\A}{K_1}{K_2}&\CNV{\A}{K_1}{\TCast{W}V}}
\\

\end{tabular}
\end{table}

\begin{theorem}[bound rt-reduction of valid terms]\
\thslabel{lsubv}
\begin{enumerate}

\item\thlabel{lsubv_cnv_trans}
\Caption{strengthening of validity through its refinement}
If $\LSubV{\A}{K}{L}$
and $\CNV{\A}{L}{T}$
then $\CNV{\A}{K}{T}$.

\item\thlabel{cnv_cpm_teqx_conf_lpr_aux}
\Caption{diamond confluence of restricted bound rt-reduction with itself under hypotheses}
If $\MALL{L_2,T_2} \FQUP{L_1}{T_1}{L_2}{T_2} \MIMP \PropD{\A}{L_2}{T_2}$
then $\PropD{\A}{L_1}{T_1}$.

\item\thlabel{cnv_cpm_teqx_cpm_trans_sub}
\Caption{diamond confluence of restricted bound rt-reduction with bound rt-reduction under hypotheses}
If $\MALL{L_2,T_2} \FPBG{L_1}{T_1}{L_2}{T_2} \MIMP \PropP{\A}{L_2}{T_2}$\\
and $\MALL{L_2,T_2} \FQUP{L_1}{T_1}{L_2}{T_2} \MIMP \PropS{\A}{L_2}{T_2}$
then $\PropS{\A}{L_1}{T_1}$.

\item\thlabel{cnv_cpm_conf_lpr_aux}
\Caption{kite confluence of bound rt-reduction with itself under hypotheses}
If $\MALL{L_2,T_2} \FPBG{L_1}{T_1}{L_2}{T_2} \MIMP \PropC{\A}{L_2}{T_2}$\\
and $\MALL{L_2,T_2} \FPBG{L_1}{T_1}{L_2}{T_2} \MIMP \PropP{\A}{L_2}{T_2}$
then $\PropK{\A}{L_1}{T_1}$.

\item\thlabel{cnv_cpm_trans_lpr_aux}
\Caption{preservation of validity by bound rt-reduction under hypotheses}
If $\MALL{L_2,T_2} \FPBG{L_1}{T_1}{L_2}{T_2} \MIMP \PropC{\A}{L_2}{T_2}$\\
and $\MALL{L_2,T_2} \FPBG{L_1}{T_1}{L_2}{T_2} \MIMP \PropP{\A}{L_2}{T_2}$
then $\PropP{\A}{L_1}{T_1}$.

\item\thlabel{cnv_cpms_conf_lpr_aux}
\Caption{full confluence of bound rt-reduction with itself under hypotheses}
If $\MALL{L_2,T_2} \FPBG{L_1}{T_1}{L_2}{T_2} \MIMP \PropC{\A}{L_2}{T_2}$\\
and $\MALL{L_2,T_2} \FPBG{L_1}{T_1}{L_2}{T_2} \MIMP \PropP{\A}{L_2}{T_2}$
then $\PropC{\A}{L_1}{T_1}$.

\item\thlabel{cnv_preserve}
\Caption{bound rt-reduction preserves validity and is confluent on valid terms}
If $\CNV{\A}{L}{T}$ then $\PropC{\A}{L}{T}$ and $\PropP{\A}{L}{T}$.

\end{enumerate}
\end{theorem}

\begin{proof}
\thref{}{lsubv_cnv_trans}
is proved by induction on the second premise
and by cases on the first premise
with the help of \thref{lprs}{lsubr_cpm_trans}
in the cases of Rule $\ruleref{cnv}{\TAppl{}}$
and Rule $\ruleref{cnv}{\TCast{}}$.
\thref{}{cnv_cpm_teqx_conf_lpr_aux} and \thref{}{cnv_cpm_teqx_cpm_trans_sub}
are proved by cases on $T_1$
with the help of \thref{lsubc}{cnv_fwd_fsb}.
\thref{}{cnv_cpm_conf_lpr_aux}
is proved by cases on $T_1$
with the help of \thref{lprs}{cpr_conf_lpr}.
\thref{}{cnv_cpm_trans_lpr_aux}
is proved by cases on $T_1$
with the help of \thref{}{lsubv_cnv_trans},
\thref{lprs}{cprs_conf} and \thref{acle}{cnv_fwd_cpms_abst_dx_le}.
\thref{}{cnv_cpms_conf_lpr_aux}
follows from \thref{}{cnv_cpm_teqx_conf_lpr_aux}, \thref{}{cnv_cpm_teqx_cpm_trans_sub},
\thref{}{cnv_cpm_conf_lpr_aux} and \thref{}{cnv_cpm_trans_lpr_aux}.
\thref{}{cnv_preserve} (the preservation theorem
and the confluence of bound rt-reduction)
is proved at once by big-tree induction on the closure $\Cl{L}{T}$
with the help of \thref{}{cnv_cpm_trans_lpr_aux}
and \thref{}{cnv_cpms_conf_lpr_aux}.
\end{proof}

\subsection{RT-Reduction to Normal Form}
\seclabel{cpmuwe}

The fact that we can rt-reduce the valid terms to suitable normal forms
is a key step towards the decidability of parametric validity (\secref{cpes}).
In particular we decide the existence of the common reducts
$W_0$ in Rule $\ruleref{cnv}{\TAppl{}}$
and $U_0$ in Rule $\ruleref{cnv}{\TCast{}}$
by comparing full r-normal forms.
Moreover, we decide if $T$ rt-reduces to a $\TAbst{\Y}{}{}$-form
in Rule $\ruleref{cnv}{\TAppl{}}$
by inspecting its weak head rt-normal form.

The normal forms we need are formally introduced next.
In \tabref{cpmre} we define the predicate $\CNR{L}{T}$,
meaning that $T$ is a full r-normal form in $L$,
and then the relation $\CPMRE{n}{L}{T_1}{T_2}$, meaning that
the term $T_1$ rt-reduces to a full r-normal form $T_2$ in $L$.
Notice that we could use single-step rt-reduction in Rule $\ruleref{cpmre}{N}$
and that this rule follows the pattern of Rule $\ruleref{csx}{N}$ (\secref{fsb}).

\begin{table}
\seccaption{cpmuwe}
{Bound rt-reduction to full r-normal form.}
\tablabel{cpmre}
\begin{tabular}{c}

\infer[\ruleref{cpmre}{N}]
{\CNR{L}{T_1}}
{\MALL{T_2}
 \MATOM{\CPMS{0}{L}{T_1}{T_2}} \MIMP
 \MATOM{\TEQ{T_1}{T_2}}
}
\sep

\infer[\ruleref{cpmre}{E}]
{\CPMRE{n}{L}{T_1}{T_2}}
{\CPMS{n}{L}{T_1}{T_2}&\CNR{L}{T_2}}
\\

\end{tabular}
\end{table}

We follow the same pattern in \tabref{cpmuwe} where we define
the relation $\CPMUWE{n}{L}{T_1}{T_2}$, meaning that
the term $T_1$ rt-reduces to a weak head rt-normal form $T_2$ in $L$.
This is based on the equivalence $\TEQW{T_1}{T_2}$ defined by the rules of \tabref{teqw}.
The reader may want to compare this relation with $\TEQX{T_1}{T_2}$ (\tabref{teqx}).
We stress that we must use multi-step rt-reduction in Rule $\ruleref{cpmuwe}{N}$,
otherwise the term $T_1 = \TAbbr{x}{\TSRef{s}}\TAbst{\Y}{y}{\TSRef{s}}{x}$
would be in weak head rt-normal form since
$\CPM{n}{L}{T_1}{X}$ implies
$X = T_2 = \TAbbr{x}{\TSRef{s}}\TAbst{\Y}{y}{\TSRef{s}}{\TSRef{s}}$
and $\TEQW{T_1}{X}$ holds because of the $\TAbst{}{}{}$-forms.
On the other hand $\CPMS{n}{L}{T_1}{X}$ admits also the solutions
$X = T_1$ and $X = T_3 = \TAbst{\Y}{y}{\TSRef{s}}{\TSRef{s}}$
that is the $\zeta$-contraction of $T_2$.

\begin{table}
\seccaption{cpmuwe}
{Equivalence for weak head rt-normal form.}
\tablabel{teqw}
\begin{tabular}{c}

\infer[\ruleref{teqw}{\TSort{}}]
{\TEQW{\TSRef{s_1}}{\TSRef{s_2}}}
{}
\sep

\infer[\ruleref{teqw}{\TLRef{}}]
{\TEQW{\TNRef{x}}{\TNRef{x}}}
{}
\sep

\infer[\ruleref{teqw}{\TAbst{}{}{}}]
{\TEQW{\TAbst{\Y}{x}{W_1}T_1}{\TAbst{\Y}{x}{W_2}T_2}}
{}
\sep

\infer[\ruleref{teqw}{\TAbbr{}{}}]
{\TEQW{\TAbbr{x}{V_1}T_1}{\TAbbr{x}{V_2}T_2}}
{\TEQW{T_1}{T_2}}
\nl

\infer[\ruleref{teqw}{\TAppl{}}]
{\TEQW{\TAppl{V}T_1}{\TAppl{V}T_2}}
{\TEQW{T_1}{T_2}}
\sep

\infer[\ruleref{teqw}{\TCast{}}]
{\TEQW{\TCast{U_1}T_1}{\TCast{U_2}T_2}}
{\TEQW{U_1}{U_2}&\TEQW{T_1}{T_2}}
\\

\end{tabular}
\end{table}

\begin{table}
\seccaption{cpmuwe}
{Bound rt-reduction to weak head rt-normal form.}
\tablabel{cpmuwe}
\begin{tabular}{c}

\infer[\ruleref{cpmuwe}{N}]
{\CNUW{L}{T_1}}
{\MALL{n,T_2}
 \MATOM{\CPMS{n}{L}{T_1}{T_2}} \MIMP
 \MATOM{\TEQW{T_1}{T_2}}
}
\sep

\infer[\ruleref{cpmuwe}{E}]
{\CPMUWE{n}{L}{T_1}{T_2}}
{\CPMS{n}{L}{T_1}{T_2}&\CNUW{L}{T_2}}
\\

\end{tabular}
\end{table}

The existence and uniqueness conditions for these normal forms 
are listed in the next theorem.
The relation $\CPES{n_1}{n_2}{L}{T_1}{T_2}$ in \thref{}{cnv_cpmuwe_mono}
is the rt-conversion that we introduce in \secref{cpes}.

\begin{theorem}[rt-reduction to normal form]\
\thslabel{cpmuwe}
\begin{enumerate}

\item\thlabel{cpms_total_aaa}
\Caption{terms with an arity rt-reduce for each $n$}
If $\AAA{L}{T_1}{A}$
then there exists $T_2$
such that $\CPMS{n}{L}{T_1}{T_2}$.

\item\thlabel{cpmre_total_aaa}
\Caption{terms with an arity rt-reduce to an r-normal form for each $n$}
If $\AAA{L}{T_1}{A}$
then there exists $T_2$
such that $\CPMRE{n}{L}{T_1}{T_2}$.

\item\thlabel{cpmuwe_total_csx}
\Caption{strongly normalizing terms rt-reduce to a weak head rt-normal form for some $n$}
If $\CSX{L}{T_1}$
then there exist $n$ and $T_2$
such that $\CPMUWE{n}{L}{T_1}{T_2}$.

\item\thlabel{cnv_cpmre_mono}
\Caption{the r-normal form valid terms rt-reduce to is unique for each $n$}
If $\CNV{\A}{L}{T}$
and $\CPMRE{n}{L}{T}{T_1}$
and $\CPMRE{n}{L}{T}{T_2}$
then $\TEQ{T_1}{T_2}$.

\item\thlabel{cnv_cpmuwe_mono}
\Caption{the weak head rt-normal form valid terms rt-reduce to is unique up to rt-conversion}
If $\CNV{\A}{L}{T}$
and $\CPMUWE{n_1}{L}{T}{T_1}$
and $\CPMUWE{n_2}{L}{T}{T_2}$\\
then $\TEQW{T_1}{T_2}$ and $\CPES{n_2-n_1}{n_1-n_2}{L}{T_1}{T_2}$.

\end{enumerate}
\end{theorem}

\begin{proof}
\thref{}{cpms_total_aaa}
is proved by induction on $n$
with the help of \thref{lsuba}{cpx_aaa_conf_lpx}
in the inductive case.
\thref{}{cpmre_total_aaa}
follows from \thref{}{cpms_total_aaa},
\thref{gcr}{csx_gcr} and \thref{lsubc}{acr_aaa_lsubc_lifts}.
\thref{}{cpmuwe_total_csx}
is proved by induction on the premise with the help of
\thref{lprs}{lpr_cpm_trans}, \thref{lsubr}{lsubr_cpm_trans}
and \thref{req}{cpx_req_conf}.
\thref{}{cnv_cpmre_mono} and \thref{}{cnv_cpmuwe_mono}
follow from \thref{lsubv}{cnv_preserve}.
\end{proof}

\subsection{RT-Conversion and Decidability of Parametric Validity}
\seclabel{cpes}

In this section we discuss the proof of the main \thref{cpes}{cnv_dec}
stating that the condition $\CNV{\A}{L}{T}$ is decidable provided that
the applicability domain $\A$ satisfies the next condition
\eqnref{ad_dec}.
\begin{equation}
\eqnlabel{ad_dec}
\hbox{For all $n_1$ we can decide
$\MEX{n_2\in\A} n_1 \le n_2$}
\end{equation}

\begin{table}
\seccaption{cpes}
{RT-conversion.}
\tablabel{cpes}
\begin{tabular}{c}

\infer[\ruleref{cpes}{I}]
{\CPES{n_1}{n_2}{L}{T_1}{T_2}}
{\CPMS{n_1}{L}{T_1}{T_0}&\CPMS{n_2}{L}{T_2}{T_0}}
\\

\end{tabular}
\end{table}

Firstly, we extend r-conversion by introducing, 
with the rule of \tabref{cpes},
the relation $\CPES{n_1}{n_2}{L}{T_1}{T_2}$ 
stating that the terms $T_1$ and $T_2$ have a common rt-reduct in $L$.
Thus r-conversion becomes $\CPES{0}{0}{L}{T_1}{T_2}$.
We use this notion to gain the next formulation of
Rule $\ruleref{cnv}{\TAppl{}}$ and Rule $\ruleref{cnv}{\TCast{}}$.
The decidability of $\CPES{n_1}{n_2}{L}{T_1}{T_2}$, stated by \thref{cpes}{cnv_cpes_dec},
stands on the comparison of $X_1$ and $X_2$ provided by 
$\CPMRE{n_1}{L}{T_1}{X_1}$ and $\CPMRE{n_2}{L}{T_2}{X_2}$ (\secref{cpmuwe}).
\begin{equation}
\eqnlabel{cnv_appl_cpes}
\vcenter{
\infer[\ruleref{cpes}{\TAppl{}}]
{\CNV{\A}{L}{\TAppl{V}T}}
{\CNV{\A}{L}{V}&\CNV{\A}{L}{T}&n\in\A&\CPMS{n}{L}{T}{\TAbst{\Y}{x}{W}U}&\CPES{1}{0}{L}{V}{W}}
}
\end{equation}
\begin{equation}
\vcenter{
\eqnlabel{cnv_cast_cpes}
\infer[\ruleref{cpes}{\TCast{}}]
{\CNV{\A}{L}{\TCast{U}T}}
{\CNV{\A}{L}{U}&\CNV{\A}{L}{T}&\CPES{0}{1}{L}{U}{T}}
}
\end{equation}

The decidability of the premise $\CPMS{n}{L}{T}{\TAbst{\Y}{x}{W}U}$
of Rule $\ruleref{cpes}{\TAppl{}}$ stands on 
$\CPMUWE{n_0}{L}{T}{X_0}$ for some $\TAbst{}{}{}$-form $X_0$
and for the least $n_0$ (\secref{cpmuwe}).
Now the premises $n\in\A$ and $\CPMS{n}{L}{T}{X}$ with $\TEQW{X}{X_0}$
can or cannot hold together depending on
$\MEX{n\in\A} n_0 \le n$, which is decidable by condition \eqnref{ad_dec}.
Notice that we can find the mentioned least $n_0$ because
the condition $\MEX{U} \CPMUWE{m}{L}{T}{U}$ is decidable for every $m$
provided that $\CNV{\A}{L}{T}$, which implies that $T$ is strongly
rt-normalizing.

\begin{theorem}[decidability of parametric validity]\
\thslabel{cpes}
\begin{enumerate}

\item\thlabel{cnv_R_cpmuwe_dec}
\Caption{for every $n$ we can decide if valid terms rt-reduce to a weak head rt-normal form}
If $\CNV{\A}{L}{T_1}$
then for every $n$ we can decide $\MEX{T_2} \CPMUWE{n}{L}{T_1}{T_2}$.

\item\thlabel{cnv_cpes_dec}
\Caption{rt-conversion of valid terms is decidable}
If $\CNV{\A}{L}{T_1}$
and $\CNV{\A}{L}{T_2}$
then we can decide $\CPES{n_1}{n_2}{L}{T_1}{T_2}$.

\item\thlabel{cnv_dec}
\Caption{validity is decidable}
if $\A$ satisfies condition \eqnref{ad_dec}
then we can decide $\CNV{\A}{L}{T}$.

\end{enumerate}
\end{theorem}

\begin{proof}
\thref{}{cnv_R_cpmuwe_dec} and \thref{}{cnv_cpes_dec}
follow from \thref{lsubv}{cnv_preserve} and \thref{cpmuwe}{cpmre_total_aaa}.
\thref{}{cnv_dec} (the decidability theorem)
is proved by induction on the closure $\Cl{L}{T}$
with the help of \thref{}{cnv_R_cpmuwe_dec}, \thref{}{cnv_cpes_dec},
\thref{cpmuwe}{cpmuwe_total_csx} and \thref{cpmuwe}{cnv_cpmuwe_mono}.
\end{proof}

\subsection{General Properties of Type Checking}
\seclabel{types}

In this section we want to discuss the clauses of \thref{types}{}
that lists the properties of the type judgment $\NTA{\A}{L}{T}{U}$
(\secref{nta}, \tabref{nta}) that do not depend on the applicability domain $\A$.

\thref{}{nta_typecheck_dec}:
asserting $\NTA{\A}{L}{T}{U}$ (type checking) corresponds to asserting
$\CNV{\A}{L}{\TCast{U}T}$ by definition, so type checking is decidable
when validity is.
\thref{}{nta_fwd_cnv_dx}:
an immediate consequence of the definition is that types are valid.
\thref{}{nta_fwd_aaa}:
another consequence is that typed terms and their types have the same arity (\secref{lsubc}), 
therefore both are strongly normalizing.
\thref{}{nta_abst_predicative}:
another implication is that  $\TAbst{\Y}{}{}$-quantification is predicative,
as we announced in \secref{introduction},
in the sense that $\NTA{\A}{L}{\TAbst{\Y}{x}{W}T}{W}$ cannot not hold.
The preservation of validity by rt-reduction (\secref{lsubv}) has three
important consequences.
\thref{}{cnv_cpms_nta}:
firstly, as we announced in \secref{cpms},
if $\CPMS{1}{L}{T}{U}$ and $T$ is valid then $U$ is indeed a type for $T$.
\thref{}{nta_fwd_cnv_sn} and \thref{}{cnv_nta_sn}:
secondly, a term is valid if and only if it has a type.
\thref{}{nta_inference_dec}:
it follows immediately that type existence is decidable when validity is.
\thref{}{nta_cprs_conf}:
thirdly, types are indeed preserved by r-reduction.
The confluence of rt-reduction on valid terms (\secref{lsubv})
yields \thref{}{nta_mono}: types are unique up to r-conversion.

Coming now to the problem of axiomatizing $\NTA{\A}{L}{T}{U}$,
we see that, apart from the case of application,
the axioms do not depend on the applicability domain $\A$
and we list them in \tabref{types}
where we recognize the \emph{start} Rule $\ruleref{types}{\TAbst{}{}{}{}}$,
the \emph{weakening} Rule $\ruleref{types}{L}$ 
and the \emph{conversion} Rule $\ruleref{types}{C}$.

We can regard the validity judgments in these rules
as shorthand for the existence of types,
which are not restricted to sorts as in a PTS.
So these types can be $\TAbst{\Y}{}{}$-abstractions as well.

\begin{table}
\seccaption{types}
{Axioms for the type judgment not concerning application.}
\tablabel{types}
\begin{tabular}{c}

\infer[\ruleref{types}{\TAbst{}{}{}{}}]
{\NTA{\A}{K\LAbst{x}{W}}{\TNRef{x}}{W}}
{\CNV{\A}{K}{W}}
\sep

\infer[\ruleref{types}{\LAbbr{}{}}]
{\NTA{\A}{K\LAbbr{x}{V}}{\TNRef{x}}{W}}
{\NTA{\A}{K}{V}{W}}
\sep

\infer[\ruleref{types}{L}]
{\NTA{\A}{K\LPair{y}{V}}{\TNRef{x}}{U}}
{\NTA{\A}{K}{\TNRef{x}}{U}}
\nl

\infer[\ruleref{types}{\TSort{}}]
{\NTA{\A}{L}{\TSRef{s}}{\TSRef{\Next{s}}}}
{}
\sep

\infer[\ruleref{types}{P}]
{\NTA{\A}{L}{\TPair{\Y}{x}{V}T}{\TPair{\Y}{x}{V}U}}
{\CNV{\A}{L}{V}&\NTA{\A}{L\LPair{x}{V}}{T}{U}}
\sep

\infer[\ruleref{types}{\TCast{}}]
{\NTA{\A}{L}{\TCast{U}T}{U}}
{\NTA{\A}{L}{T}{U}}
\nl

\infer[\ruleref{types}{C}]
{\NTA{\A}{L}{T}{U_2}}
{\NTA{\A}{L}{T}{U_1}&\CPES{0}{0}{L}{U_1}{U_2}&\CNV{\A}{L}{U_2}}
\nl

Rule $\ruleref{types}{L}$:
$y \neq x_1$ and $y$ not free in $T_2$.\\

\end{tabular}
\end{table}

\begin{theorem}[general properties of expected type checking]\
\thslabel{types}
\begin{enumerate}

\item\thlabel{nta_fwd_cnv_sn}
\Caption{typed terms are valid}
If $\NTA{\A}{L}{T}{U}$
then $\CNV{\A}{L}{T}$.

\item\thlabel{nta_fwd_cnv_dx}
\Caption{expected types are valid, correctness of types}
If $\NTA{\A}{L}{T}{U}$
then $\CNV{\A}{L}{U}$.

\item\thlabel{nta_fwd_aaa}
\Caption{typed terms and their expected types have the same arity}
If $\NTA{\A}{L}{T}{U}$
then there exists $A$
such that $\AAA{L}{T}{A}$
and $\AAA{L}{U}{A}$.

\item\thlabel{nta_abst_predicative}
\Caption{universal abstraction is predicative}
$\NTA{\A}{L}{\TAbst{\Y}{x}{W}T}{W}$
does not hold.

\item\thlabel{cnv_cpms_nta}
\Caption{the inferred types of valid terms are indeed expected types of such terms}
If $\CNV{\A}{L}{T}$
and $\CPMS{1}{L}{T}{U}$
then $\NTA{\A}{L}{T}{U}$.

\item\thlabel{cnv_nta_sn}
\Caption{valid terms have an expected type}
If $\CNV{\A}{L}{T}$
then there exists $U$
such that $\NTA{\A}{L}{T}{U}$.

\item\thlabel{nta_cprs_conf}
\Caption{r-reduction preserves the expected type, subject reduction}
If $\NTA{\A}{L}{T_1}{U}$
and $\CPMS{0}{L}{T_1}{T_2}$
then $\NTA{\A}{L}{T_2}{U}$.

\item\thlabel{nta_mono}
\Caption{the expected types of a term are unique up to r-conversion}
If $\NTA{\A}{L}{T}{U_1}$
If $\NTA{\A}{L}{T}{U_2}$
then $\CPES{0}{0}{L}{U_1}{U_2}$.

\item\thlabel{nta_typecheck_dec}
\Caption{expected type checking is decidable}
if $\A$ satisfies condition \eqnref{ad_dec}
then we can decide $\NTA{\A}{L}{T}{U}$.

\item\thlabel{nta_inference_dec}
\Caption{expected type inference is decidable}
if $\A$ satisfies condition \eqnref{ad_dec}
then we can decide $\MEX{U} \NTA{\A}{L}{T}{U}$.

\end{enumerate}
\end{theorem}

\begin{proof}
\thref{}{nta_fwd_cnv_sn} and
\thref{}{nta_fwd_cnv_dx} (correctness of types)
are immediate.
\thref{}{nta_fwd_aaa}
follows from \thref{lsubc}{cnv_fwd_aaa}.
\thref{}{nta_abst_predicative}
follows from \thref{}{nta_fwd_aaa} and \thref{lsuba}{aaa_mono}.
\thref{}{cnv_cpms_nta},
\thref{}{nta_cprs_conf} (subject reduction)
and \thref{}{nta_mono} (uniqueness of types) 
follow from \thref{lsubv}{cnv_preserve}.
\thref{}{cnv_nta_sn}
follows from \thref{}{cnv_cpms_nta}.
\thref{}{nta_typecheck_dec}
follows from \thref{cpes}{cnv_dec}.
\thref{}{nta_inference_dec}
follows from \thref{}{nta_fwd_cnv_sn}, \thref{}{cnv_nta_sn}
and \thref{cpes}{cnv_dec}.
\end{proof}

\thref{types}{nta_fwd_aaa} tells us that arity,
\ie the functional structure of terms, is preserved by typing.
Thus each type, including each sort,
determines the functional structure of its inhabitants.  
In this consideration lies the deep reason why 
our system cannot support the BHK-based PAT interpretation of logic, \ie
in this interpretation the sort of propositions does not have the mentioned property.

\subsection{Iterated Type Checking for Terms}
\seclabel{ntas}

As of the axioms for typing the application, we notice that
Rule $\ruleref{cnv}{\TAppl{}}$ (\secref{nta})
makes an essential use of the bound $n$,
which appears in the premise $n\in\A$.
If $\A$ is fixed in advance, we can reasonably find a set of axioms
that do not mention $n$ (see for example \appref{ld2a} and \appref{ld1a}),
but in the parametric case we need to define
a generalized type judgment in which $n$ occurs.
This is the relation $\NTAS{\A}{n}{L}{T_0}{T_n}$,
which defines $T_n$ as an expected $n$-type of $T_0$
and holds when we can prove a chain of type checks
$\NTA{\A}{L}{T_i}{T_{i+1}}$ where $0 \le i < n$.
Therefore $\NTAS{\A}{1}{L}{T}{U}$ means $\NTA{\A}{L}{T}{U}$.
Formally speaking, we introduce this notion with the rule of \tabref{ntas},
that generalizes Rule $\ruleref{nta}{\TCast{}}$ by replacing the constant $1$ with $n$
and that immediately gives $\NTAS{\A}{n}{L}{T}{U}$
if $\CNV{\A}{L}{T}$ and $\CPMS{n}{L}{T}{U}$.
This generalization of \thref{types}{cnv_cpms_nta}
states that the inferred $n$-types of valid terms are indeed
expected $n$-types of such terms and is proved in the same way.

\begin{table}
\seccaption{ntas}
{Iterated type checking for terms.}
\tablabel{ntas}
\begin{tabular}{c}

\infer[\ruleref{ntas}{\TCast{}}]
{\NTAS{\A}{n}{L}{T}{U}}
{\CNV{\A}{L}{U}&\CNV{\A}{L}{T}&\CPMS{n}{L}{T}{U_0}&\CPMS{0}{L}{U}{U_0}}
\\

\end{tabular}
\end{table}

Not surprisingly,
we can axiomatize $\NTAS{\A}{n}{L}{T}{U}$
in terms of $\NTA{\A}{L}{T}{U}$
with the next three rules.
\begin{equation}
\eqnlabel{ntas_zero}
\vcenter{
\infer[\ruleref{ntas}{R}]
{\NTAS{\A}{0}{L}{T_1}{T_2}}
{\CNV{\A}{L}{T_1}&\CNV{\A}{L}{T_2}&\CPES{0}{0}{L}{T_1}{T_2}}
}\sep
\vcenter{
\infer[\ruleref{ntas}{I}]
{\NTAS{\A}{1}{L}{T_1}{T_2}}
{\NTA{\A}{L}{T_1}{T_2}}
}\sep
\vcenter{
\infer[\ruleref{ntas}{T}]
{\NTAS{\A}{n_1+n_2}{L}{T_1}{T_2}}
{\NTAS{\A}{n_1}{L}{T_1}{T}&\NTAS{\A}{n_2}{L}{T}{T_2}}
}
\end{equation}
As we see, $\NTAS{\A}{n}{L}{T}{U}$ comprises
convertibility, validity and type checking in a single notion. 

This machinery provides for the next equivalent formulation
of Rule $\ruleref{cnv}{\TAppl{}}$ mentioning types.
\begin{equation}
\eqnlabel{cnv_appl_ntas}
\vcenter{
\infer[\ruleref{ntas}{\TAppl{}}]
{\CNV{\A}{L}{\TAppl{V}T}}
{n\in\A&\NTA{\A}{L}{V}{W}&\NTAS{\A}{n}{L}{T}{\TAbst{\Y}{x}{W}U}}
}
\end{equation}

By this rule, the parametric axioms for typing the application can take the form:
\begin{equation}
\eqnlabel{nta_appl_ntas_zero}
\vcenter{
\infer[\ruleref{ntas}{\TAppl{}0}]
{\NTA{\A}{L}{\TAppl{V}T}{\TAppl{V}\TAbst{\Y}{x}{W}U}}
{0\in\A&\NTA{\A}{L}{V}{W}&\NTAS{\A}{0}{L}{T}{\TAbst{\Y}{x}{W}U_0}&\NTA{\A}{L\LAbst{x}{W}}{U_0}{U}}
}
\end{equation}
\begin{equation}
\eqnlabel{nta_appl_ntas_pos}
\vcenter{
\infer[\ruleref{ntas}{\TAppl{}+}]
{\NTA{\A}{L}{\TAppl{V}T}{\TAppl{V}U}}
{m+1\in\A&\NTA{\A}{L}{V}{W}&\NTA{\A}{L}{T}{U}&\NTAS{\A}{m}{L}{U}{\TAbst{\Y}{x}{W}U_0}}
}
\end{equation}
where we need to make a distinction between
the cases $n = 0$ and $n = m + 1$, \ie $n > 0$.

\subsection{Partial Order for Applicability Domains}
\seclabel{acle}

In this section we present a general condition $\ACLE{\A_1}{\A_2}$
on two applicability domains
by which $\CNV{\A_1}{L}{T}$ implies $\CNV{\A_2}{L}{T}$
for every closure $\Cl{L}{T}$.
The condition is defined by the rule of \tabref{acle}.

\begin{table}
\seccaption{acle}
{Partial order for applicability domains.}
\tablabel{acle}
\begin{tabular}{c}

\infer[\ruleref{acle}{I}]
{\ACLE{\A_1}{\A_2}}
{\MALL{m\in\A_1}\MEX{n\in\A_2} m \le n}
\\

\end{tabular}
\end{table}

The proof of our claim is straightforward and corresponds to \thref{acle}{cnv_acle_trans}.
For example
the following are immediate relationships between some applicability domains
used in \appref{members}:
$\ACLE{\ACOne}{\ACZeroOne}$, $\ACLE{\ACZeroOne}{\ACOne}$,
$\ACLE{\ACZero}{\ACZeroOne}$ and $\ACLE{\ACZeroOne}{\SysOmega}$
where $\SysOmega \defeq \ACAny$.

Obviously, we can prove that $\ACLE{\A_1}{\A_2}$
is a reflexive and transitive relation as one expects.

\begin{theorem}[partial order for applicability domains]\
\thslabel{acle}
\begin{enumerate}

\item\thlabel{cnv_fwd_cpms_abst_dx_le}
\Caption{weakening of bound rt-reduction from a valid term to a universal abstraction}
If $\CNV{\A}{L}{T}$
and $\CPMS{n_1}{L}{T}{\TAbst{\Y}{x}{W}U_1}$
and $n_1 \le n_2$\\
then there exists $U_2$
such that $\CPMS{n_2}{L}{T}{\TAbst{\Y}{x}{W}U_2}$
and $\CPMS{n_2-n_1}{L\LAbst{x}{W}}{U_1}{U_2}$.

\item\thlabel{cnv_acle_trans}
\Caption{weakening of validity through partial order}
If $\CNV{\A_1}{L}{T}$ and $\ACLE{\A_1}{\A_2}$ then $\CNV{\A_2}{L}{T}$.

\end{enumerate}
\end{theorem}

\begin{proof}
\thref{}{cnv_fwd_cpms_abst_dx_le}
comes from \thref{lprs}{lpr_cpm_trans},
\thref{lsubc}{cnv_fwd_aaa} and \thref{cpmuwe}{cpms_total_aaa}.
\thref{}{cnv_acle_trans}
is proved by induction on the first premise
with the help of \thref{}{cnv_fwd_cpms_abst_dx_le}
in the case of Rule $\ruleref{cnv}{\TAppl{}}$.
\end{proof}

\section{Conclusions}
\seclabel{conclusion}

In this article we presented a system that we term hereafter $\LD{2}{B}$,
which includes and evolves two earlier systems:
$\LD{1}{A}$ of \citeN{lambdadeltaJ1a} and $\LD{2}{A}$ of \citeN{lambdadeltaR2c}.
In regard to their names, $\TAbst{}{}{}$ and $\TAbbr{}{}$ denote
the main binders in their syntax,
while 1A, 2A, 2B identify a major and a minor version.

We stress that the systems of the $\LD{}{}$ family are not
related intentionally to any other system
having (variations of) the symbols $\lambda$ and $\delta$ in its name or syntax.
Examples include, but are not limited to:
$\ChLD$ \cite{Chu41},
$\DL$ \cite{Bru93},
$\lambda_\Delta$ \cite{RS94},
$\PLC$ \cite{RP04},
$\lambda D$ \cite{NG14},
$\LFDelta$ \cite{HLSS18},
$\CLX$ \cite{Bru78},
and others like \cite{Ned79,Ned80,BKN96}.

As $\LD{2}{B}$ is not intended as a logical framework,
we are not suggesting an interpretation of logic into it.
This would mean discussing the notion of proposition, assertion and proof. 

We advocate two original ideas:
on the one hand,
introducing bound rt-reduction (\secref{cpms})
to define our system
and, on the other hand,
introducing sort irrelevance (\secref {feqx})
to avoid degree-based inductions, and the notion of degree in itself,
in the system's theory we developed.

Our main contributions are, on the one hand,
a shorter definition for the system with respect to $\LD{2}{A}$ and, on the other hand,
some important results that we are presenting here for the first time.
The definition (\secref{definition}) 
stands just on two notions, \ie bound rt-reduction and parametric validity.
Our new results are the confluence of bound rt-reduction
on valid terms (\secref{lsubv}), the decidability of validity (\secref{cpes}),
of type inference and type checking (\secref{types})
and an axiomatization of the type judgment for $\LD{2}{A}$ (\appref{ld2a}).
Some results already proved for $\LD{2}{A}$ are included for completeness:
confluence of r-reduction (\secref{lprs}),
strong normalization of qrst-reduction (\secref{lsubc}) and
preservation of validity by bound rt-reduction (\secref{lsubv}).

Moreover,
we specified the whole theory of $\LD{2}{B}$ in the \CoIC{} and
we checked it with the interactive prover Matita (\appref{specification}).
This process took five years.

We want to stress that
the presented framework is a family of systems depending on a parameter $\A$
that rules the so-called applicability condition
(see Rule $\ruleref{cnv}{\TAppl{}}$ of \secref{nta}).
As we show in \appref{members}, both $\LD{1}{A}$ and $\LD{2}{A}$ emerge naturally
by setting this parameter to specific values. 
There, in \appref{rpce}
we conjecture a link between these systems by means of $\eta$-conversion.

Moreover, in \appref{exclusion} we report on the exclusion binder,
a technical device appearing in $\LD{1}{A}$ (\cite{lambdadeltaJ1a}, Section 4).
In the rest of this section we discuss the key features of the
systems in the $\LD{}{}$ family comparing then with similar frameworks.
One remark resulting from the discussion is that our systems borrow
some features from the pure type systems (PTS) \cite{Brn92} 
and some from the languages of the Automath tradition,
but stand outside both families.

Future work on the $\LD{}{}$ family includes proving
the conjecture on $\eta$-conversion we pose in \appref{rpce}
and generalizing the type annotation $\TCast{U}T$ with the
iterated type annotation $\TICast{n}{U}{T}$,
whose validity condition is given by the next Rule $\ruleref{cnv}{\TICast{n}{}}$
(to be compared with Rule $\ruleref{ntas}{\TCast{}}$).
\begin{equation}
\eqnlabel{icast}
\vcenter{
\infer[\ruleref{cnv}{\TICast{n}{}}]
{\CNV{\A}{L}{\TICast{n}{U}T}}
{\CNV{\A}{L}{U}&\CNV{\A}{L}{T}&\CPMS{n}{L}{T}{U_0}&\CPMS{0}{L}{U}{U_0}}
}
\end{equation}
With this device we can define $\NTAS{\A}{n}{L}{T}{U}$
as $\CNV{\A}{L}{\TICast{n}{U}T}$ superseding Rule
$\ruleref{nta}{\TCast{}}$ and Rule $\ruleref{ntas}{\TCast{}}$.
The reader should be aware that reducing the number of auxiliary
notions needed to develop the theory of the $\LD{}{}$ family is a
priority for us and that we will continue to work in this direction.
On the one hand, we are interested in the conditions that allow to regard
our r-conversion as definitional equality in the sense of
\citeN{ML75}. Our objective is to propose one of our systems as
a specification language for the Minimalist Foundation of \citeN{Mai09}.
See the online Appendix A of \citeN{lambdadeltaJ1a} for more details.
On the other hand we are interested in the theory of the layered
abstraction $\TAbst{e}{x}{W}T$ of \citeN{lambdadeltaJ3a}.
Our objective is to develop a system, term it $\LYP$,
that solves some issues with $\AutQE$ \cite{SPAa3},
\ie the logical framework of the Automath family joining $\LY$ of \citeN{SPAc6} and $\LP$.
In particular we expect $\LYP$ to maintain the desired properties of both systems,
especially confluence of r-reduction and uniqueness of types up to r-conversion.
In this situation naively using $\TAbst{\Y}{}{}$ in place of $\TProd{}{}{}$
is not the right solution since $\TAbst{\Y}{}{}$ cannot satisfy
the introductions rules for $\TProd{}{}{}$ that in $\LP$ follow the pattern:
\begin{equation}
\eqnlabel{prod}
\vcenter{
\infer[\mathrm{(product)}]
{\NTA{}{L}{\TProd{}{x}{W}T}{U_2}}
{\NTA{}{L}{W}{U_1}&\NTA{}{L\LAbst{x}{W}}{T}{U_2}}
}
\end{equation}
In $\LYP$ we expect to validate a suitable translation of
the \emph{Grundlagen} \cite{Jut79}, \ie
the only non-trivial development based on de Bruijn's quantifier
validated first in $\AutQE$.

\subsection{Notes on Sorts}
\seclabel{sort}

We designed $\LD{2}{B}$ to avoid the distinction between valid terms and
typable terms by making sure that every valid term has a type.
This distinction, occurring in systems like the $\LCube$
where the sort $\SortB$ is valid but untyped,
complicates the statement and the proof of many system's properties,
and yet is useless since we can inject any system with top sorts
into a system without top sorts.

Moreover $\LD{2}{B}$ enforces uniqueness of types up to r-conversion,
thus a sort determines its type.  

We write a generic sort as $\TSRef{s}$,
where $s$ is a meta-variable ranging over a set $\Sort$ of identifiers.

The type of $\TSRef{s}$ in the $\LD{}{}$ family is $\TSRef{\Next{s}}$
where $\Next{}$ is a parameter function from $\Sort$ to $\Sort$.

Restrictions on $\Sort$ and $\Next{}$ appear in $\LD{1}{A}$ and $\LD{2}{A}$,
but both $\Sort$ and $\Next{}$ can be chosen at will in $\LD{2}{B}$  
provided that $\Sort \neq \emptyset$ and equality in $\Sort$ is decidable.
As observed by \citeN{Web20b} for $\Sd$, 
even choosing $\Sort \defeq \{\TSort{}\}$ and $\Next{\TSort{}} \defeq \TSort{}$
does not break the properties of the system.

In particular valid terms are strongly normalizing regardless of the
chosen sort structure because sorts do not type $\TAbst{\Y}{}{}$-expressions.
This is opposed to PTS's where sorts do type $\TProd{}{}{}$-expressions. 

Certainly, many predicates we are presenting in this article
take $\Sort$ and $\Next{}$ as parameters,
but for the reader's convenience we hide them
in the notation we are proposing for these predicates.

\subsection{Notes on Variables}
\seclabel{variable}

As we show in Rule $\ruleref{types}{\TAbst{}{}{}{}}$
of \secref{types} and line \eqnref{start},
the expected type of a valid declared variable, \ie $W$,
must be valid in its environment.
but it does not need to be a sort nor to have a sort type.
\begin{equation}
\eqnlabel{start}
\vcenter{
\infer[\ruleref{types}{\TAbst{}{}{}{}}]
{\NTA{\A}{K\LAbst{x}{W}}{\TNRef{x}}{W}}
{\CNV{\A}{K}{W}}
}
\end{equation}
In this way, and contrary to a PTS, type chains like
$\CTypeV{L}{\TNRef{x_1}}{\TNRef{x_n}}{W}$ are possible for any $n$ 
as in some Automath-related systems.
See the $\AutIV$ family of \citeN{SPAb3} for applications.

In $\LD{2}{B}$ all variables are bound locally.
However, global constants are unavoidable in realistic mathematical developments
and we plan to support them in future extensions of our system.

\subsection{Notes on Abstraction and Quantification}
\seclabel{abstraction}

Universal abstraction, which we write $\TAbst{\Y}{x}{W}$,
is the construction for abstracting a variable $\TNRef{x}$ of type $W$ in the
system $\DL$ \cite{Bru93} and in simpler systems of the Automath
tradition such as: $\SLambda$ \cite{Ned73}, $\VrLL$ \cite{SPAc4}, $\LY$ \cite{SPAc6}.
As designed by de Bruijn, this construction is characterized by the next requirements
mentioned in \secref{introduction}.
\begin{enumerate}

\item
$\TAppl{V}\TAbst{\Y}{x}{W}T$ is a $\beta$-redex
(in our system: Rule $\ruleref{cpm}{\beta}$ of \secref{cpms});

\item
$\TAbst{\Y}{x}{W}T$ is a weak head normal form (see \citeN{PJ87} 11.3.1);

\item
a type of $\TAbst{\Y}{x}{W}T$ has the form $\TAbst{\Y}{x}{W}U$
(in our system: Rule $\ruleref{types}{P}$ of \secref{types}).

\end{enumerate}

Outside the project Automath, we find it in $\GrLL$ \cite{Gro93}
and in $\Sd$ \cite{Web20b}.

Universal abstraction arises naturally when we consider the
simplest syntactical modification that turns untyped $\lambda$-calculus
into Church-typed $\lambda$-calculus, \ie
adding at least one sort, decorating $\lambda$-abstractions with
an expected type and taking the resulting terms as the types.
The well-known result of this modification is a $\lambda$-calculus with
uniformly dependent types, whose valid terms (and types) are also valid in
a system of simple types referred to as norms or arities.

We can explain universal abstraction as follows: 
the multiple type judgment 
$\CTypeI{L\LAbst{x}{W}}{T}{U_1}{U_2}$ under the condition $\CNV{}{L}{W}$
gives
$\CTypeI{L}{\TAbst{\Y}{x}{W}T}{\TAbst{\Y}{x}{W}U_1}{\TAbst{\Y}{x}{W}U_2}$
meaning that the function $\TAbst{\Y}{x}{W}T$
belongs to the dependent function space $\TAbst{\Y}{x}{W}U_1$,
which belongs to the family of function spaces $\TAbst{\Y}{x}{W}U_2$,
all whose members have $W$ as their source. 

The importance of $\TAbst{\Y}{}{}$-abstraction appears clearly
from the \emph{propositions-as-types} (PAT) perspective
(see for example \cite{KLN04})
since it provides for a predicative universal quantification on sorts.
In fact $\CType{L}{\TAbst{\Y}{x}{W}T}{W}$ does not hold 
regardless of the term $W$, including any sort.

Higher-order universal quantifications,
for instance on the sort of propositions,
are unavoidable in formal logic
(inference rules do contain propositional variables)
and a predicativist logician accepts them provided that
they are schematic, \ie they occur just in the meta-language.

When such quantifications are formalized explicitly in typed
$\lambda$-calculus following the PAT interpretation of logic,
the predicativist logician expects to render them with predicative means,
and $\TAbst{\Y}{}{}$-abstraction is one of such means indeed.
Moreover we notice that the $\TAbst{\Y}{}{}$-quantification of schematic variables
is standard practice in the languages of the Automath family from the start.

Other constructions serving the same purpose are
the parametric constants of $\CpPTS$ \cite{KLN01} and
the parametric abstraction and quantification ($\S$ and $\P$) of $\LLXVIII$ \cite{KLN03}.
In this respect, we can see $\TAbst{\Y}{}{}$
as a construction playing the role of both $\S$ and $\P$.

Some systems of the Automath family enforce another well-known
identification: the one concerning functional abstraction and universal quantification
(the constructions $\lambda$ and $\Pi$ of the $\LCube$).
If in $\Aut$ \cite{SPAb1} this identification is just syntactical, 
for we can recover the role of each unified binder by looking at its degree,
the situation is different is $\AutQE$ \cite{SPAa3},
where the unified binders of degree $2$ can play both roles at once.

For the reader's convenience we recall that a \emph{degree} 
is an integer indicating the position of a valid term in a type hierarchy.
Its definition depends on the specific type system, but it is always
the case that a term has degree $d$ when its type has degree $d-1$.
For example in the systems related to the $\LCube$
the definition assigns degree $0$ to sort $\SortB$.
Thus every term of type $\SortB$, including sort $\SortA$, has degree $1$
and so on.
On the other hand, the Automath tradition assigns degree $1$ to its sort $\tau$.

Outside the Automath family, the unification of $\lambda$ and $\Pi$
is investigated by \citeN{Kam05}.

Although discussing $\lambda$ and $\Pi$ goes beyond the scope of this article,
we notice that $\TAbst{\Y}{}{}$ can unify just the $\TAbst{}{}{}{}$
and the $\TProd{}{}{}$ of category $\ShapeSS$
according to the classification of \citeN{Brn92}.
Thus, we stress that the $\LD{}{}$-family does not aim at
the unification of $\TAbst{}{}{}{}$ and $\TProd{}{}{}$
in the general sense.

\subsection{Notes on Application}
\seclabel{application}

As its predecessors, $\LD{2}{B}$ displays application according to the
so-called item notation of \citeN{KN96b} in order to improve the visual
understanding of redexes.
In particular, $\TAppl{V}T$ denotes
the function $T$ applied to the argument $V$,
whose mainstream notation is $\Appl{(V)}{T}$. 

The condition on which an application is valid,
\ie the so-called \emph{applicability} condition,
deserves some general comments.
An important observation concerns the case
$\CNV{}{L}{\TAppl{V}\TNRef{x_1}}$.
In a PTS-like type system, where the types of types are sorts or
(by our mistake!) in $\LD{1}{A}$,
we just require that $\TNRef{x_1}$ inhabits a function space,
for instance by postulating $\CType{L}{\TNRef{x_1}}{\TProd{}{x}{W}U}$
or $\CType{L}{\TNRef{x_1}}{\TAbst{\Y}{x}{W}U}$ (provided that $\CType{L}{V}{W}$).
On the other hand, in an Automath-like type system as $\AutQE$ or $\LD{2}{A}$,
where type chains of the kind $\CTypeV{L}{\TNRef{x_1}}{\TNRef{x_n}}{\cdots}$ are possible,
we must extend the previous condition by postulating that there
exists some $n$ for which $\CTypeV{L}{\TNRef{x_1}}{\TNRef{x_n}}{\TAbst{\Y}{x}{W}U}$.

In order to catch both the restricted (PTS-like) condition
and the extended (Automath-like) condition,
$\LD{2}{B}$ allows to choose $n$ in a parametric subset $\A$ of natural numbers.
So $\A = \SUBSET{\hbox{any}\ n}{}$ implies the extended condition,
$\A = \SUBSET{1}{}$ implies the restricted condition,
if $\A = \SUBSET{0}{}$ only the applications of explicit functions are valid,
\ie $\TAppl{V}\TAbst{\Y}{x}{W}T$,
if $\A = \SUBSET{}{}$ no application is valid. 

The applicability condition in Automath tradition
allows $\CNV{}{L}{\TAppl{V}\TAbst{\Y}{x}{W}T}$ of every degree
which allows $\CNV{}{L}{\TAbst{\Y}{x}{W}T}$
(provided that $\CType{L}{V}{W}$).
The reader should notice that $T$ can be a type and, in this case,
$\TAbst{\Y}{x}{W}T$ is a dependent function space.
The benefit of this choice is that every typable term $X$ has a unique canonical type $\Typ{X}$
whose construction is syntax-directed and does not involve $\beta\zeta\theta\epsilon$-reduction.
Interestingly, \citeN{KBN99} show that the same thing works in the $\LCube$ extended with
$\Pi$-application, $\Pi$-reduction and definitions in environments. 

Comforted by these results, the approach of $\LD{2}{B}$ is to
allow $\CNV{}{L}{\TAppl{V}\TAbst{\Y}{x}{W}T}$ of any degree.

Finally we want to stress that
the applicability condition in this article is strong in that
$\CNV{}{L}{\TAppl{V}T}$ implies $\CNV{}{L}{V}$ and $\CNV{}{L}{T}$.
This is opposed to the weak applicability of $\DL$ \cite{Bru93},
where $\CNV{}{L}{\TAppl{V}T}$ implies $\CNV{}{L}{V}$ but not
necessarily $\CNV{}{L}{T}$.
In particular, the weak validity of $\TAppl{V}\TAbst{\Y}{x}{W}T$
corresponds to the strong validity of the $\beta$-reduct $\TAbbr{x}{\TCast{W}{V}}T$
which does not imply the validity of $\TAbst{\Y}{x}{W}T$.
Notice that $\DL$ does not have $\TAbbr{x}{\TCast{W}{V}}T$ and uses
$\TAppl{V}\TAbst{\Y}{x}{W}T$ instead.
On the contrary, we distinguish these constructions.
For the reader's convenience, here is a term that
is weakly valid but not strongly valid.
For $\CTypeI{L}{V}{W_1}{W_2}$ and $\CType{L}{T}{U}$,
take $\TAppl{W_1}\TAbst{}{y}{W_2}\TAppl{V}\TAbst{}{x}{y}T$.
The subterm $\TAbst{}{y}{W_2}\TAppl{V}\TAbst{}{x}{y}T$
is not valid because of the $\beta$-redex,
but the whole term is valid because so is its outer $\beta$-reduct
$\TAbbr{y}{W_1}\TAppl{V}\TAbst{}{x}{y}T$.

\subsection{Notes on Local Definitions}
\seclabel{substitution}

The $\LD{}{}$ family,
aiming to be a realistic tool for the specification of mathematics,
features non-recursive local definitions in terms
because of their undeniable importance in the practice of mathematical development.
Indeed, mathematics is unimaginable without some support for definitions.
For instance they appear in functional programming languages,
in practical implementations of \CoIC{} \cite{CoqArt}
and of other logical frameworks.

In particular the construction $\TAbbr{x}{V}T$ represents the
definition of $\TNRef{x}$ as $V$ in $T$
and allows the $\LD{}{}$ family to delay the substitution of $V$ in $T$
during $\beta$-reduction, that is realized in call-by-name style.

We notice that a sequence of adjacent definitions,
like $\TAbbr{x_1}{V_1} \ldots \TAbbr{x_n}{V_n}T$,
realizes what is known in the Automath tradition as a telescopic
(or dependent) explicit substitution,
\ie a multiple explicit substitution in which every term $V_i$ may depend on
each variable $x_j$ such that $j < i$.

We also notice that the presence of definitions allows
to develop the whole meta-theory of $\LD{2}{B}$
(and of $\LD{2}{A}$ and contrary to $\LD{1}{A}$)
without the meta-linguistic substitution, \ie $\Subst{x}{V}{T}$.

\subsection{Notes on Expected Type Annotation}
\seclabel{annotation}

Another linguistic feature that the $\LD{}{}$ family takes from real-world
programming languages and logical frameworks
is the expected type annotation, which is also known as the explicit type cast.

This construction appears as $\TCast{U}T$ in the present article
and may appear as $\Cast{U}{T}$ elsewhere.
Both notations mean that $U$ is an expected type of $T$.
Its importance lies in the possibility of expressing a type checking
problem as a validity asserting problem, \ie  
$\CType{L}{T}{U}$ iff $\CNV{}{L}{\TCast{U}T}$.

Moreover the variant of $\TCast{}$ termed $\TICast{0}{}$,
that we will define with Rule \eqnref{icast}
in the forthcoming systems of the $\LD{}{}$ family,
will allow to express with validity even the convertibility of valid terms.

\subsection{Notes on Environments}
\seclabel{environment}

In the $\LD{}{}$ family an environment is the structure containing the
information on the free variables of a valid term.
In the literature it may appear as a context,
but we follow \citeN{SU06} that use this term for a different structure.
There are well-established motivations
to allow more than just type declarations in environments
and in the systems of the $\LD{}{}$ family we allow non-recursive definitions as well.
As in the system of \citeN{KBN99},
this feature enables small-step $\beta$-reduction,
which in turn helps to achieve type checking by separating
type construction and type conversion
while the desired properties are maintained (see \appref{cpts}).

Some systems, see for instance $\LY$ of \citeN{SPAc6} and $\MF$ of \citeN{Mai09},
provide for the notion of valid environment, \ie $\CValid{}{L}$,
and are designed so as to ensure that $\CValid{L}{T}$ implies $\CValid{}{L}$.
If on the one hand this notion simplifies the proofs of some
system's properties and for instance is reasonable in the scenario of
categorical semantics (where objects correspond to valid environments),
on the other hand it is redundant in many cases and
forces a mutual recursion between the definitions of $\CValid{}{L}$
and $\CValid{L}{T}$ that the systems of the $\LD{}{}$ family aim to avoid.

\begin{acks}

I am grateful to C. Sacerdoti Coen and S. Solmi
for their support and valuable advice on the matter of this article.
I thank the anonymous referees
for their helpful suggestions with which I could improve this text.
I wish to dedicate this work to a very special lady,
T. Ilie, for her constant closeness
and for the joyful moments we shared in these years
during the development of the $\LD{}{}$ systems.

\end{acks}

\bibliography{%
lambdadelta,%
spa,%
kamareddine,%
helm,%
inductives,%
mtt,%
accattoli,%
coq,%
references%
}

\newpage

\appendix

\section{Some Members of the Framework}
\seclabel{members}

As we saw in \secref{nta}, the definition of our framework depends on a
parameter $\A$ possibly satisfying condition \eqnref{ad_dec}.
Thus, what we defined is a system $\LD{2}{B}\ad\A$ for each specific choice of $\A$, 

In \secref{acle} we developed a general tool for relating these systems,
by which we argue that some of them are equivalent,
\ie they validate the same closures $\Cl{L}{T}$.
Three choices for $\A$ deserve special attention here.
Firstly,
the system $\LD{2}{B}\ad\SysOmega$, where $\SysOmega = \ACAny$,
is essentially the system $\LD{2}{A}$ of \citeN{lambdadeltaR2c} (\secref{ld2a}).
Secondly,
the system $\LD{2}{B}\ad\ACOne$, which is equivalent to $\LD{2}{B}\ad\ACZeroOne$,
is essentially the system $\LD{1}{A}$ of \citeN{lambdadeltaJ1a} (\secref{ld1a}).
Finally,
the system $\LD{2}{B}\ad\ACZero$ comes into play
when we try to relate $\LD{2}{B}\ad\SysOmega$ with $\LD{2}{B}\ad{\SUBSET{1}{}}$
using $\eta$-conversion (\secref{rpce}).

\subsection{Iterated Inferred Type Assignment for Terms}
\seclabel{cpts}

Generally speaking and following the widely accepted terminology of \citeN{Cos96},
the inferred (or preferred) type of a term $T$ is a specific representative in the
equivalence class of the types of $T$, which are defined up to r-conversion.
It is provided by the function $\Typ{}$ of the Automath tradition \cite{Dln80}
and it is the canonical type of \citeN{KN96a} in a PTS.
It appears as the static type in $\LD{1}{A}$ and in $\LD{2}{A}$,
and plays a central role in the latter system being part of its definition.
Here, we introduce it as an auxiliary notion for the sole purpose of
relating $\LD{2}{A}$ with $\LD{2}{B}\ad\SysOmega$ in \secref{ld2a}
and we make some changes with respect to the former presentations.
Firstly, we do not define a single inferred type of $T$,
but a family of alternatives that comprise the (different) choices
we made in $\LD{1}{A}$ and in $\LD{2}{A}$.
Secondly, we do not require an inferred type of $T$ in $L$ to be closed in $L$.
This approach brings a simplification and works without loss of generality
since an inferred type $U$ of $T$ in $L$ is indeed a type of $T$ in $L$ under 
the assumption that $T$ is valid in $L$, which already ensures that $U$ is closed in $L$.
Notice that in $\LD{2}{A}$ the inferred type $U$ must be closed in $L$
so as to ensure that $U$ has a degree in $L$,
but, thanks to sort irrelevance of \secref{feqx},
we saw that the theory of $\LD{2}{B}$ can stand without the notion of degree.  

In the current setting an inferred $n$-type of $T_1$ is 
a term $T_2$ we obtain from $T_1$ by iterating type inference $n$ times
in the sense of \secref{cpms}
(this is related to $\NTAS{\A}{n}{L}{T_1}{T_2}$ of \secref{ntas}). 

In particular we define a relation $\CPT{n}{L}{T_1}{T_2}$
by picking the rules of $\CPM{n}{L}{T_1}{T_2}$ (\tabref{cpm}, \secref{cpms})
that make it sensible to consider $T_2$ as the inferred $n$-type of $T_1$ in $L$.
These are the rules of \tabref{cpt}.
What is missing from bound rt-reduction are the $\beta\zeta\theta\epsilon$-steps.
In other words we consider just the t-steps and the $\delta$-steps.
It follows by construction that $\CPT{n}{L}{T_1}{T_2}$ implies
$\CPM{n}{L}{T_1}{T_2}$, thus by \thref{types}{cnv_cpms_nta},
$\CNV{\A}{L}{T}$ and $\CPT{1}{L}{T}{U}$ imply $\NTA{\A}{L}{T}{U}$
as one expects.

\begin{table}
\appcaption{cpts}
{Bound t-reduction for terms (one step).}
\tablabel{cpt}
\begin{tabular}{c}

\infer[\ruleref{cpt}{L}]
{\CPT{n}{K\LPair{y}{V}}{\TNRef{x}}{T}}
{\CPT{n}{K}{\TNRef{x}}{T}}
\sep

\infer[\ruleref{cpt}{\delta}]
{\CPT{0}{K\LAbbr{x}{V}}{\TNRef{x}}{V}}
{}
\sep

\infer[\ruleref{cpt}{l}]
{\CPT{1}{K\LAbst{x}{W}}{\TNRef{x}}{W}}
{}
\nl

\infer[\ruleref{cpt}{e}]
{\CPT{1}{L}{\TCast{U}T}{U}}
{}
\sep

\infer[\ruleref{cpt}{Pl}]
{\CPT{0}{L}{\TPair{\Y}{x}{V_1}T}{\TPair{\Y}{x}{V_2}T}}
{\CPT{0}{L}{V_1}{V_2}}
\sep

\infer[\ruleref{cpt}{Pr}]
{\CPT{n}{L}{\TPair{\Y}{x}{V}T_1}{\TPair{\Y}{x}{V}T_2}}
{\CPT{n}{L\LPair{x}{V}}{T_1}{T_2}}
\nl

\infer[\ruleref{cpt}{s}]
{\CPT{1}{L}{\TSRef{s}}{\TSRef{\Next{s}}}}
{}
\sep

\infer[\ruleref{cpt}{\TAppl{}l}]
{\CPT{0}{L}{\TAppl{V_1}T}{\TAppl{V_2}T}}
{\CPT{0}{L}{V_1}{V_2}}
\sep

\infer[\ruleref{cpt}{\TAppl{}r}]
{\CPT{n}{L}{\TAppl{V}T_1}{\TAppl{V}T_2}}
{\CPT{n}{L}{T_1}{T_2}}
\nl

\infer[\ruleref{cpt}{\TCast{}l}]
{\CPT{0}{L}{\TCast{U_1}T}{\TCast{U_2}T}}
{\CPT{0}{L}{U_1}{U_2}}
\sep

\infer[\ruleref{cpt}{\TCast{}r}]
{\CPT{0}{L}{\TCast{U}T_1}{\TCast{U}T_2}}
{\CPT{0}{L}{T_1}{T_2}}
\sep

\infer[\ruleref{cpt}{\TCast{}b}]
{\CPT{1}{L}{\TCast{U_1}T_1}{\TCast{U_2}T_2}}
{\CPT{1}{L}{U_1}{U_2}&\CPT{1}{L}{T_1}{T_2}}
\nl

Rule $\ruleref{cpt}{L}$:
$y \neq x_1$ and $y$ not free in $T_2$.\\

\end{tabular}
\end{table}

Then we define $\CPTS{n}{L}{T_1}{T_2}$ following the pattern of $\CPMS{n}{L}{T_1}{T_2}$
with the rules of \tabref{cpts}.
This is the general inferred $n$-type assignment
whose desired properties are listed in \thref{cpts}{}.

Notice that we can gain more invariants by strengthening the
inferred type assignment.
These may include the uniqueness and closeness properties we
required both in $\LD{1}{A}$ and $\LD{2}{A}$.

\begin{theorem}[iterated inferred type assignment and bound rt-reduction]\
\thslabel{cpts}
\begin{enumerate}

\item\thlabel{cpts_cpms_conf_eq}
\Caption{iterated inferred types and bound rt-reducts are r-convertible on valid terms}
If $\CNV{\A}{L}{T_0}$
and $\CPTS{n}{L}{T_0}{T_1}$
and $\CPMS{n}{L}{T_0}{T_2}$
then $\CPES{0}{0}{L}{T_1}{T_2}$.

\item\thlabel{cpts_cprs_trans}
\Caption{iterated type inference and r-reduction composed as bound rt-reduction}
If $\CPTS{n}{L}{T_1}{T}$ and $\CPMS{0}{L}{T}{T_2}$
then $\CPMS{n}{L}{T_1}{T_2}$.

\end{enumerate}
\end{theorem}

\begin{proof}
\thref{}{cpts_cpms_conf_eq}
follows from \thref{cpes}{cnv_dec}.
\thref{}{cpts_cprs_trans}
is immediate.
\end{proof}

\subsection{A System with Automath-Like Applicability}
\seclabel{ld2a}

In this section we argue that
the system $\LD{2}{B}\ad\SysOmega$ we introduced in this article
is equivalent to the system $\LD{2}{A}$ of \citeN{lambdadeltaR2c}.
This system has the same constructions, rt-reduction rules and validity
rules (because the premise $n\in\SysOmega$ of Rule $\ruleref{cnv}{\TAppl{}}$
always holds) apart from the fact that bound rt-reduction is replaced
by iterated type inference (\secref{cpts}) followed by r-reduction
in Rule $\ruleref{cnv}{\TAppl{}}$ and Rule $\ruleref{cnv}{\TCast{}}$.
Our main point here is that \thref{cpts}{} guarantees the equivalence of
these reduction sequences, given that we apply these sequences just to valid terms.

Moreover, we can provide for the first time a set of axioms that 
fully describe the type relation of $\LD{2}{B}\ad\SysOmega$.
They are the rules of \tabref{types} and the next.
Rule $\ruleref{types}{\TAppl{}\omega}$ is advocated by \citeN{Bru91}.
\begin{equation}
\eqnlabel{nta_2a}
\vcenter{
\infer[\ruleref{types}{\TAppl{}0}]
{\NTA{\SysOmega}{L}{\TAppl{V}\TAbst{\Y}{x}{W}T}{\TAppl{V}\TAbst{\Y}{x}{W}U}}
{\NTA{\SysOmega}{L}{V}{W}&\NTA{\SysOmega}{L\LAbst{x}{W}}{T}{U}}
}\sep
\vcenter{
\infer[\ruleref{types}{\TAppl{}\omega}]
{\NTA{\SysOmega}{L}{\TAppl{V}T}{\TAppl{V}U}}
{\NTA{\SysOmega}{L}{T}{U}&\CNV{\SysOmega}{L}{\TAppl{V}U}}
}
\end{equation}

\subsection{A System with PTS-Like Applicability}
\seclabel{ld1a}

In this section we argue that
the system $\LD{2}{B}\ad\ACOne$ we introduced in this article
is equivalent to the system $\LD{1}{A}$ of \citeN{lambdadeltaJ1a}.
This system has the same terms and additional environment constructions,
which are unused in conversion and typing.
The reductions are the same except for $\beta$-contraction,
$\CStep{L}{\TAppl{V}\TAbst{\Y}{x}{W}T}{}{\beta}{\TAbbr{x}{V}T}$, 
in which the expected type $W$ is lost.
The single r-reduction step allows for less parallelism,
but the r-reduction sequences reach the same normal forms,
thus convertibility is preserved.
Since we specified $\LD{1}{A}$ and $\LD{2}{B}\ad\ACOne$ in different
versions of $\CIC$ (see \appref{specification}),
we are not able to verify this equivalence formally at the moment.

The main point here is that the type relation of $\LD{1}{A}$ is primitive
and we can prove that its axioms fully describe the type relation of $\LD{2}{B}\ad\ACOne$.
They are the rules of \tabref{types} and the next:
\begin{equation}
\eqnlabel{nta_appl_old}
\vcenter{
\infer[\ruleref{types}{\TAppl{}1}]
{\NTA{\ACOne}{L}{\TAppl{V}T}{\TAppl{V}\TAbst{\Y}{x}{W}U}}
{\NTA{\ACOne}{L}{V}{W}&\NTA{\ACOne}{L}{T}{\TAbst{\Y}{x}{W}U}}
}
\end{equation}

In $\LD{1}{A}$ we use the next type rule for $\TCast{U}T$,
which is equivalent to Rule $\ruleref{types}{\TCast{}}$.
\begin{equation}
\eqnlabel{nta_cast_old}
\vcenter{
\infer[\ruleref{types}{\TCast{}p}]
{\NTA{\ACOne}{L}{\TCast{U}T}{\TCast{W}U}}
{\NTA{\ACOne}{L}{T}{U}&\NTA{\ACOne}{L}{U}{W}}
}
\end{equation}

These considerations lead us to reserve a special name for the system
$\LD{2}{B}\ad\ACOne$, \ie $\LD{1}{B}$. 

Interestingly, the observations in \secref{acle}
imply the equivalence of $\LD{2}{B}\ad\ACOne$ and $\LD{2}{B}\ad\ACZeroOne$.

Finally, we remark that in this system, as in a PTS,
a structural induction on the type judgment is enough to prove
the preservation of type by r-reduction, thus the big-tree theorem is not needed.

\subsection{A Conjecture on Eta-Conversion}
\seclabel{rpce}

In this section we formulate a conjecture about the relation between
$\LD{2}{B}\ad\SysOmega$ and $\LD{2}{B}\ad\ACZero$.
Informally speaking, if a closure $\Cl{L}{T}$ is valid in $\LD{2}{B}\ad\SysOmega$,
then a suitable $\eta$-expansion of it is valid in $\LD{2}{B}\ad\ACZero$.
This $\eta$-expansion, which concerns the term $T$ and the terms in $L$
recursively referred by $T$, must be applied systematically on every
variable instance $x$ referring to an abstraction $\TAbst{\Y}{y}{W}$
or to a declaration $\LAbst{y}{W}$
where the expected type $W$ is functional, \ie rt-reduces to a function.

Formally, we apply the $\eta$-expansion to $\Cl{L_1}{T_1}$
with two functional relations.
One on terms: $\CPCE{L_1}{T_1}{T_2}$ (\tabref{cpce})
and one on environments $\RPCE{T}{L_1}{L_2}$ (\tabref{rpce}).
Both relations are total when $\CNV{\A}{L_1}{T_1}$ and
our expectation is:
$\MATOM{\CNV{\SysOmega}{L_1}{T_1}} \MIMP
 \MATOM{\RPCE{T_1}{L_1}{L_2}} \MIMP
 \MATOM{\CPCE{L_1}{T_1}{T_2}} \MIMP
 \MATOM{\CNV{\ACZero}{L_2}{T_2}}
$.

\begin{table}
\appcaption{rpce}
{$\eta$-expansion on declared variable occurrences for terms (one parallel step).}
\tablabel{cpce}
\begin{tabular}{c}

\infer[\ruleref{cpce}{\TSort{}}]
{\CPCE{L}{\TSRef{s}}{\TSRef{s}}}
{}
\sep

\infer[\ruleref{cpce}{L}]
{\CPCE{K\LPair{y}{V}}{\TNRef{x_1}}{T_2}}
{\CPCE{K}{\TNRef{x_1}}{T_2}}
\sep

\infer[\ruleref{cpce}{\LAbbr{}{}}]
{\CPCE{K\LAbbr{x}{V}}{\TNRef{x}}{x}}
{}
\nl

\infer[\ruleref{cpce}{\TAbst{}{}{}{}}]
{\CPCE{K\LAbst{x}{V}}{\TNRef{x}}{x}}
{\MALL{n,W,U}\NCPMS{n}{K}{V}{\TAbst{\Y}{y}{W}U}}
\sep

\infer[\ruleref{cpce}{\eta}]
{\CPCE{K\LAbst{x}{V}}{\TNRef{x}}{\TAbst{\Y}{y}{W}\TAppl{y}x}}
{\CPMS{n}{K}{V}{\TAbst{\Y}{y}{W}U}}
\nl

\infer[\ruleref{cpce}{P}]
{\CPCE{L}{\TPair{\Y}{x}{V_1}T_1}{\TPair{\Y}{x}{V_2}T_2}}
{\CPCE{L}{V_1}{V_2}&\CPCE{L\LPair{x}{V_1}}{T_1}{T_2}}
\sep

\infer[\ruleref{cpce}{F}]
{\CPCE{L}{\TFlat{V_1}T_1}{\TFlat{V_2}T_2}}
{\CPCE{L}{V_1}{V_2}&\CPCE{L}{T_1}{T_2}}
\nl

Rule $\ruleref{cpce}{L}$:
$y \neq x$ and $y$ not free in $T_2$.\sep
Rule $\ruleref{cpce}{\eta}$:
$y \neq x$ and $x$ not free in $W$.\\

\end{tabular}
\end{table}

\begin{table}
\appcaption{rpce}
{$\eta$-expansion on declared variable occurrences for environments (on selected entries).}
\tablabel{rpce}
\begin{tabular}{c}

\infer[\ruleref{rpce}{\LAtom}]
{\RPCE{\f}{\LAtom}{\LAtom}}
{}
\sep

\infer[\ruleref{rpce}{B}]
{\RPCE{\f}{K_1\LPair{y}{V_1}}{K_2\LPair{y}{V_2}}}
{\RPCE{\f}{K_1}{K_2}}
\nl

\infer[\ruleref{rpce}{P}]
{\RPCE{\f\SOR\SUBSET{y}{}}{K_1\LPair{y}{V_1}}{K_2\LPair{y}{V_2}}}
{\RPCE{\f}{K_1}{K_2}&\CPCE{K_1}{V_1}{V_2}}
\nl

Rules $\ruleref{rpce}{B}$ and $\ruleref{rpce}{P}$: $y \notin \f$.
\sep
$\RPCE{T}{L_1}{L_2}$ means $\RPCE{\FreeP{L_1}{T}}{L_1}{L_2}$.
\\

\end{tabular}
\end{table}

The work on the proof is in progress and seems to involve
some notions, as standard rt-reduction sequences,
whose exposition deserves more space than the one
at our disposal in this article.

We stress that
\citeN{lambdadeltaJ3a} supports our conjecture
by showing that the \emph{Grundlagen} can be translated
from an extension of $\LD{2}{B}\ad\SysOmega$
to the corresponding extension of $\LD{2}{B}\ad\ACOne$,
which contains $\LD{2}{B}\ad\ACZero$ as of \secref{acle},
by $\eta$-expanding 21 declared variable instances as
we do in Rule $\ruleref{cpce}{\eta}$.

\section{Auxiliary definitions}

In this appendix we list the auxiliary definitions
we did not include in the body of the article.

\begin{table}[!ht]
\seccaption{fqus}
{Proper subclosure (sequence of s-steps).}
\tablabel{fqup}
\begin{tabular}{c}

\infer[\ruleref{fqup}{I}]
{\FQUP{L_1}{T_1}{L_2}{T_2}}
{\FQU{L_1}{T_1}{L_2}{T_2}}
\sep

\infer[\ruleref{fqup}{T}]
{\FQUP{L_1}{T_1}{L_2}{T_2}}
{\FQUP{L_1}{T_1}{L}{T}&\FQUP{L}{T}{L_2}{T_2}}
\\

\end{tabular}
\end{table}

\begin{table}[!ht]
\seccaption{fqus}
{Reflexive subclosure (reflexive sequence of s-steps).}
\tablabel{fqus}
\begin{tabular}{c}

\infer[\ruleref{fqus}{R}]
{\FQUS{L}{T}{L}{T}}
{}
\sep

\infer[\ruleref{fqus}{I}]
{\FQUS{L_1}{T_1}{L_2}{T_2}}
{\FQU{L_1}{T_1}{L_2}{T_2}}
\sep

\infer[\ruleref{fqus}{T}]
{\FQUS{L_1}{T_1}{L_2}{T_2}}
{\FQUS{L_1}{T_1}{L}{T}&\FQUS{L}{T}{L_2}{T_2}}
\\

\end{tabular}
\end{table}

\begin{table}[!ht]
\seccaption{cpxs}
{Extended rt-reduction for terms (one step).}
\tablabel{cpx}
\begin{tabular}{c}

Reduction rules and type inference rules (r-steps and t-steps)\nl

\infer[\ruleref{cpx}{\beta}]
{\CPX{L}{\TAppl{V}\TAbst{\Y}{x}{W}T}{\TAbbr{x}{\TCast{W}V}T}}
{}
\sep

\infer[\ruleref{cpx}{\theta}]
{\CPX{L}{\TAppl{V}\TAbbr{x}{W}T}{\TAbbr{x}{W}\TAppl{V}T}}
{}
\nl

\infer[\ruleref{cpx}{\#}]
{\CPX{K\LPair{x}{V}}{\TNRef{x}}{V}}
{}
\sep

\infer[\ruleref{cpx}{\zeta}]
{\CPX{L}{\TAbbr{x}{V}T}{T}}
{}
\sep

\infer[\ruleref{cpx}{\epsilon}]
{\CPX{L}{\TCast{U}T}{T}}
{}
\nl

\infer[\ruleref{cpx}{s}]
{\CPX{L}{\TSRef{s_1}}{\TSRef{s_2}}}
{}
\sep

\infer[\ruleref{cpx}{e}]
{\CPX{L}{\TCast{U}T}{U}}
{}
\nl

Context rules\nl

\infer[\ruleref{cpx}{L}]
{\CPX{K\LPair{y}{V}}{\TNRef{x}}{T}}
{\CPX{K}{\TNRef{x}}{T}}
\sep

\infer[\ruleref{cpx}{\TAppl{}l}]
{\CPX{L}{\TAppl{V_1}T}{\TAppl{V_2}T}}
{\CPX{L}{V_1}{V_2}}
\sep

\infer[\ruleref{cpx}{\TAppl{}r}]
{\CPX{L}{\TAppl{V}T_1}{\TAppl{V}T_2}}
{\CPX{L}{T_1}{T_2}}
\nl

\infer[\ruleref{cpx}{Pl}]
{\CPX{L}{\TPair{\Y}{x}{V_1}T}{\TPair{\Y}{x}{V_2}T}}
{\CPX{L}{V_1}{V_2}}
\sep

\infer[\ruleref{cpx}{Pr}]
{\CPX{L}{\TPair{\Y}{x}{V}T_1}{\TPair{\Y}{x}{V}T_2}}
{\CPX{L\LPair{x}{V}}{T_1}{T_2}}
\nl

\infer[\ruleref{cpx}{\TCast{}l}]
{\CPX{L}{\TCast{U_1}T}{\TCast{U_2}T}}
{\CPX{L}{U_1}{U_2}}
\sep

\infer[\ruleref{cpx}{\TCast{}r}]
{\CPX{L}{\TCast{U}T_1}{\TCast{U}T_2}}
{\CPX{L}{T_1}{T_2}}
\nl

Rule $\ruleref{cpx}{\zeta}$:
$x$ not free in $T$.
\sep
Rule $\ruleref{cpx}{L}$:
$y \neq x$ and $y$ not free in $T$.
\\

\end{tabular}
\end{table}

\begin{table}[!ht]
\seccaption{cpxs}
{Extended rt-reduction for terms (sequence of steps).}
\tablabel{cpxs}
\begin{tabular}{c}

\infer[\ruleref{cpxs}{R}]
{\CPXS{L}{T}{T}}
{}
\sep

\infer[\ruleref{cpxs}{I}]
{\CPXS{L}{T_1}{T_2}}
{\CPX{L}{T_1}{T_2}}
\sep

\infer[\ruleref{cpxs}{T}]
{\CPXS{L}{T_1}{T_2}}
{\CPXS{L}{T_1}{T}&\CPXS{L}{T}{T_2}}
\\

\end{tabular}
\end{table}

\begin{table}[!ht]
\seccaption{lprs}
{R-reduction for environments (sequence of steps on all entries).}
\tablabel{lprs}
\begin{tabular}{c}

\infer[\ruleref{lprs}{I}]
{\LPRS{L_1}{L_2}}
{\LPR{L_1}{L_2}}
\sep

\infer[\ruleref{lprs}{T}]
{\LPRS{L_1}{L_2}}
{\LPRS{L_1}{L}&\LPRS{L}{L_2}}
\\

\end{tabular}
\end{table}

\begin{table}[!ht]
\seccaption{lpxs}
{Extended rt-reduction for environments (one step on all entries).}
\tablabel{lpx}
\begin{tabular}{c}

\infer[\ruleref{lpx}{B}]
{\LPX{K_1\LPair{y}{V}}{K_2\LPair{y}{V}}}
{\LPX{K_1}{K_2}}
\sep

\infer[\ruleref{lpx}{P}]
{\LPX{K\LPair{y}{V_1}}{K\LPair{y}{V_2}}}
{\CPX{K}{V_1}{V_2}}
\\

\end{tabular}
\end{table}

\begin{table}[!ht]
\seccaption{lpxs}
{Extended rt-reduction for environments (sequence of steps on all entries).}
\tablabel{lpxs}
\begin{tabular}{c}

\infer[\ruleref{lpxs}{R}]
{\LPXS{L}{L}}
{}
\sep

\infer[\ruleref{lpxs}{I}]
{\LPXS{L_1}{L_2}}
{\LPX{L_1}{L_2}}
\sep

\infer[\ruleref{lpxs}{T}]
{\LPXS{L_1}{L_2}}
{\LPXS{L_1}{L}&\LPXS{L}{L_2}}
\\

\end{tabular}
\end{table}

\begin{table}[!ht]
\seccaption{fpbs}
{Extended qrst-reduction for closures (sequence of steps).}
\tablabel{fpbs}
\begin{tabular}{c}

\infer[\ruleref{fpbs}{I}]
{\FPBS{L_1}{T_1}{L_2}{T_2}}
{\FPB{L_1}{T_1}{L_2}{T_2}}
\sep

\infer[\ruleref{fpbs}{T}]
{\FPBS{L_1}{T_1}{L_2}{T_2}}
{\FPBS{L_1}{T_1}{L}{T}&\FPBS{L}{T}{L_2}{T_2}}
\\

\end{tabular}
\end{table}

\begin{table}[!ht]
\seccaption{gcr}
{Neutral (or simple) term.}
\tablabel{simple}
\begin{tabular}{c}

\infer[\ruleref{simple}{\TSRef{}}]
{\Simple{\TSRef{s}}}
{}
\sep

\infer[\ruleref{simple}{\TLRef{}}]
{\Simple{\TLRef{x}}}
{}
\sep

\infer[\ruleref{simple}{F}]
{\Simple{\TFlat{V}T}}
{}
\\

\end{tabular}
\end{table}

\begin{table}[!ht]
\appcaption{cpts}
{Bound t-reduction for terms (sequence of steps).}
\tablabel{cpts}
\begin{tabular}{c}

\infer[\ruleref{cpts}{R}]
{\CPTS{0}{L}{T}{T}}
{}
\sep

\infer[\ruleref{cpts}{I}]
{\CPTS{n}{L}{T_1}{T_2}}
{\CPT{n}{L}{T_1}{T_2}}
\sep

\infer[\ruleref{cpts}{T}]
{\CPTS{n_1+n_2}{L}{T_1}{T_2}}
{\CPTS{n_1}{L}{T_1}{T}&\CPTS{n_2}{L}{T}{T_2}}
\\

\end{tabular}
\end{table}

\section{Certified Specification}
\seclabel{specification}

Contrary to a common practice,
we developed the theory of $\LD{2}{B}$ from the start
as the machine-checked specification of \citeN{lambdadeltaV2b},
which is not the formalized counterpart of some previous informal material.
As we can see in \tabref{specification},
the proofs of the few main results presented in this article
break down to 2000 lemmas that were not practical to handle just with pen and paper.

\begin{table}[!ht]
\caption
{Summary of the certified specification (October 2015 to September 2020).}
\tablabel{specification}
\begin{tabular}{lrr}
\toprule
Branch&
Definitions&
Propositions\\
\midrule
Shared structures for the $\LD{}{}$ family&
142&
989\\
Specific structures for $\LD{2}{B}$&
46&
1065\\
\bottomrule
\end{tabular}
\end{table}

The specification currently exists only in its digital version
and is developed within the {\CoIC} ($\CIC$) of \citeN{CP90}
with the help of the interactive prover Matita of \citeN{ARST11}.
This article \emph{informalizes} some selected definitions and proofs.
The corresponding proof objects are available in full as resources of the
Hypertextual Electronic Library of Mathematics (HELM) of \citeN{APSGS03}
(see \appref{pointers}).

Information on the current status of the $\LD{}{}$ family
is available at \DURL{\LDHome}.

The reader should be aware that the specification differs from the
theory presented here in some minor respects. 
For example, we represent variable occurrences with
position indexes by depth \cite{SPAc2} rather than with names
in order to turn $\alpha$-equivalence into syntactical equivalence.
Moreover, we use parallel reduction \cite{Tak95} rather than sequential reduction
in order to reduce the number of rules in some definitions and, thus,
the number of cases in some proofs.
In addition, we consider excluded entries in environments
as we explain in the next \appref{exclusion}.

\subsection{Environments with Excluded Entries}
\seclabel{exclusion}

\begin{flushright}
--- Why don't you show your face to your king?\\
--- Sire, because I do not exist!\\
Italo Calvino, The Nonexistent Knight\\
\end{flushright}

Even if free variable occurrences are unlikely to appear in the theory
of a typed $\lambda$-calculus, where all terms of interest are typed
and thus are closed in their environment, there are cases in which they do occur.
In our case we must consider Rule $\ruleref{rsx}{P\dx}$ of \secref{fsb} line $\eqnref{rsxP}$
and next line $\eqnref{rsxPdx}$,
where the environment of $T$ in the premise is $L\LPair{x}{V}$
and $x$ may be free in $T$,
while the environment of $T$ in the conclusion is just $L$
and $x$ is not bound by $L$.
This is an issue in the certified specification of $\LD{2}{B}$
because we refer to a variable by position via its depth index according to \citeN{SPAc2}.
In this situation the variable references of $T$
in the premise of the rule and in its conclusion are not related
by the well-established functions introduced by \citeN{SPAc2}.
\begin{equation}
\eqnlabel{rsxPdx}
\vcenter{
\infer[\ruleref{rsx}{P\dx}]
{\RSX{T}{L}}
{\RSX{\TPair{\Y}{x}{V}T}{L}}
}
\end{equation}

Certainly we could solve the issue by changing the way we refer to
variables, or we could set up an ad-hoc correlation function, or we
could even bind $x$ to a fake declaration in the conclusion of the rule,
but the most elegant solution to us lies on considering
the exclusion binder of \citeN{lambdadeltaJ1a}:
a device we removed from $\LD{2}{A}$, but that
the ongoing discussion fully justifies.

In particular we extend environments with the clause $L \GDEF K\LVoid{x}$
and we pose that an occurrence of $x$ formally bound by $\LVoid{x}$ is free.
The letter $\chi$ is taken after $\chi\acute{\alpha}o\sigma$:
Greek for \emph{gaping void}.

This extension requires to add specific rules
that take care of the entry $\LVoid{x}$.
This leads to pose that $\LPair{y}{V}$ includes $\LVoid{y}$
in the $L$-rules and in the $B$-rules.
Then we need to introduce the $X$-rules of \tabref{X}
and the definition:
$\FreeP{K\LVoid{x}}{\TNRef{x}} \defeq \SUBSET{x}{}$
$\MATOM{\ruleref{freep}{X}}$.
In the end Rule $\ruleref{rsx}{P\dx}$ looks as follows.
\begin{equation}
\eqnlabel{rsxXdx}
\vcenter{
\infer[\ruleref{rsx}{X\dx}]
{\RSX{T}{L\LVoid{x}}}
{\RSX{\TPair{\Y}{x}{V}T}{L}}
}
\end{equation}

\begin{table}
\appcaption{exclusion}
{Rules for environments with excluded entries.}
\tablabel{X}
\begin{tabular}{c}

\infer[\ruleref{lpr}{X}]
{\LPR{K_1\LVoid{y}}{K_2\LVoid{y}}}
{\LPR{K_1}{K_2}}
\sep

\infer[\ruleref{lpx}{X}]
{\LPX{K_1\LVoid{y}}{K_2\LVoid{y}}}
{\LPX{K_1}{K_2}}
\sep

\infer[\ruleref{jsx}{X}]
{\JSX{K_1\LPair{y}{V}}{K_2\LVoid{y}}}
{\JSX{K_1}{K_2}&\RSX{K_2}{V}}
\nl

\infer[\ruleref{req}{X}]
{\REQ{\f\SOR\SUBSET{y}{}}{K_1\LVoid{y}}{K_2\LVoid{y}}}
{\REQ{\f}{K_1}{K_2}}
\sep

\infer[\ruleref{reqx}{X}]
{\REQX{\f\SOR\SUBSET{y}{}}{K_1\LVoid{y}}{K_2\LVoid{y}}}
{\REQX{\f}{K_1}{K_2}}
\sep

\infer[\ruleref{rpce}{X}]
{\RPCE{\f\SOR\SUBSET{y}{}}{K_1\LVoid{y}}{K_2\LVoid{y}}}
{\RPCE{\f}{K_1}{K_2}}
\nl

Rule $\ruleref{jsx}{X}$ replaces Rule $\ruleref{jsx}{P}$.
\sep
Rules $\ruleref{req}{X}$, $\ruleref{reqx}{X}$, $\ruleref{rpce}{X}$: $y \notin \f$.
\\

\end{tabular}
\end{table}

The reader should notice that in this extension
equation $\eqnref{freep}$ of \secref{freep} takes the next form.
\begin{equation}
\eqnlabel{freepX}
\FreeP{L}{\TPair{\Y}{x}{V}T} =
\FreeP{L}{V} \SOR \FreeP{L\LVoid{x}}{T}
\SDIFF \SUBSET{x}{}
\end{equation}

Observing that the excluded entry $\LVoid{x}$ always appears at the right-hand
side of an environment, \ie not in the middle of it, we could
avoid it by referring to $x$ by level rather than by depth.
This observation leads us quite naturally to pose the general question whether
the theory of a typed $\lambda$-calculus is formalized more conveniently
by referring to variables by level or by depth.

\subsection{Pointers to the Certified Specification}
\seclabel{pointers}

At the moment of writing this article,
the certified specification of $\LD{2}{B}$ is available on the Web at
\DURL{http://helm.cs.unibo.it/lambdadelta/download/lambdadelta_2B.tar.bz2}
as a bundle of script files
for the i.t.p. Matita version 0.99.4.
For each proposition stated in this article
we give a pointer consisting of a path with four components:
a two-level directory inside the bundle,
a file name inside this directory
and a proved statement inside this file.
Notice that the notation in the files and in the article
may differ because of incompatibilities between
the characters available for \LaTeX{} and for Matita.
Moreover, we might modify these pointers in the forthcoming revisions of $\LD{2}{B}$.
Here we are referring to
the revision 2020-12-08 19:00:50
of the directory \DURL{/matita/matita/contribs/lambdadelta/}
of the \verb+helm.git+ repository
at \DURL{http://matita.cs.unibo.it/gitweb/}.
\begin{itemize}

\pointer{lsubr}{lsubr_cpm_trans}{\BII}{\RTTran}{cpm\_lsubr}{lsubr\_cpm\_trans}
\pointer{lsubr}{lsubr_cpx_trans}{\BII}{\RTTran}{cpx\_lsubr}{lsubr\_cpx\_trans}

\pointer{lprs}{lpr_cpm_trans}{\BII}{\RTComp}{cpms\_lpr}{lpr\_cpm\_trans}
\pointer{lprs}{cpr_conf_lpr}{\BII}{\RTTran}{lpr\_lpr}{cpr\_conf\_lpr}
\pointer{lprs}{cprs_conf}{\BII}{\RTComp}{cprs\_cprs}{cprs\_conf}
\pointer{lprs}{lpr_conf}{\BII}{\RTTran}{lpr\_lpr}{lpr\_conf}
\pointer{lprs}{lprs_conf}{\BII}{\RTComp}{lprs\_lprs}{lprs\_conf}

\pointer{lpxs}{fqu_cpx_trans}{\BII}{\RTTran}{cpx\_fqus}{fqu\_cpx\_trans}
\pointer{lpxs}{lpx_fqu_trans}{\BII}{\RTTran}{lpx\_fquq}{lpx\_fqu\_trans}
\pointer{lpxs}{lpx_cpx_trans}{\BII}{\RTComp}{cpxs\_lpx}{lpx\_cpx\_trans}

\pointer{req}{req_fqu_trans}{\SII}{\Static}{reqg\_fqus}{reqg\_fqu\_trans}
\pointer{req}{cpx_req_conf_sn}{\BII}{\RTTran}{cpx\_reqg}{cpx\_reqg\_conf\_sn}
\pointer{req}{cpx_req_conf}{\BII}{\RTTran}{rpx\_reqg}{cpx\_teqg\_repl\_reqg}
\pointer{req}{lpx_req_conf}{\BII}{\RTTran}{rpx\_reqg}{rpx\_reqg\_conf}

\pointer{fpbs}{fpbs_inv_star}{\BII}{\RTComp}{fpbs\_lpxs}{fpbs\_inv\_star}
\pointer{fpbs}{fpbs_intro_star}{\BII}{\RTComp}{fpbs\_lpxs}{fpbs\_intro\_star}


\pointer{fsb}{rsx_cpx_trans_jsx}{\BII}{\RTComp}{jsx\_rsx}{rsx\_cpx\_trans\_jsx}
\pointer{fsb}{rsx_lref_pair_lpxs}{\BII}{\RTComp}{rsx\_csx}{rsx\_lref\_pair\_lpxs}
\pointer{fsb}{csx_rsx}{\BII}{\RTComp}{rsx\_csx}{csx\_rsx}
\pointer{fsb}{csx_fsb_fpbs}{\BII}{\RTComp}{fsb\_csx}{csx\_fsb\_fpbs}

\pointer{gcr}{cpxs_fwd_beta}{\BII}{\RTComp}{cpxs\_teqo}{cpxs\_fwd\_beta}
\pointer{gcr}{csx_gcr}{\BII}{\RTComp}{csx\_gcr}{csx\_gcr}
\pointer{gcr}{acr_gcr}{\SII}{\Static}{gcp\_cr}{acr\_gcr}

\pointer{lsuba}{aaa_mono}{\SII}{\Static}{aaa\_aaa}{aaa\_mono}
\pointer{lsuba}{aaa_dec}{\SII}{\Static}{aaa\_dec}{aaa\_dec}
\pointer{lsuba}{lsuba_aaa_conf}{\SII}{\Static}{lsuba\_aaa}{lsuba\_aaa\_conf}
\pointer{lsuba}{lsuba_aaa_trans}{\SII}{\Static}{lsuba\_aaa}{lsuba\_aaa\_trans}
\pointer{lsuba}{cpx_aaa_conf_lpx}{\BII}{\RTTran}{lpx\_aaa}{cpx\_aaa\_conf\_lpx}

\pointer{lsubc}{cnv_fwd_aaa}{\BII}{\Dynamic}{cnv\_aaa}{cnv\_fwd\_aaa}
\pointer{lsubc}{acr_aaa_lsubc_lifts}{\SII}{\Static}{gcp\_aaa}{acr\_aaa\_lsubc\_lifts}
\pointer{lsubc}{cnv_fwd_fsb}{\BII}{\Dynamic}{cnv\_fsb}{cnv\_fwd\_fsb}

\pointer{lsubv}{lsubv_cnv_trans}{\BII}{\Dynamic}{lsubv\_cnv}{lsubv\_cnv\_trans}
\pointer{lsubv}{cnv_cpm_teqx_conf_lpr_aux}{\BII}{\Dynamic}{cnv\_cpm\_teqx\_conf}{cnv\_cpm\_teqx\_conf\_lpr\_aux}
\pointer{lsubv}{cnv_cpm_teqx_cpm_trans_sub}{\BII}{\Dynamic}{cnv\_cpm\_teqx\_trans}{cnv\_cpm\_teqx\_cpm\_trans\_sub}
\pointer{lsubv}{cnv_cpm_conf_lpr_aux}{\BII}{\Dynamic}{cnv\_cpm\_conf}{cnv\_cpm\_conf\_lpr\_aux}
\pointer{lsubv}{cnv_cpm_trans_lpr_aux}{\BII}{\Dynamic}{cnv\_cpm\_trans}{cnv\_cpm\_trans\_lpr\_aux}
\pointer{lsubv}{cnv_cpms_conf_lpr_aux}{\BII}{\Dynamic}{cnv\_cpms\_conf}{cnv\_cpms\_conf\_lpr\_aux}
\pointer{lsubv}{cnv_preserve}{\BII}{\Dynamic}{cnv\_preserve}{cnv\_preserve}

\pointer{cpmuwe}{cpms_total_aaa}{\BII}{\RTComp}{cpms\_aaa}{cpms\_total\_aaa}
\pointer{cpmuwe}{cpmre_total_aaa}{\BII}{\RTComp}{cpmre\_aaa}{cpmre\_total\_aaa}
\pointer{cpmuwe}{cpmuwe_total_csx}{\BII}{\RTComp}{cpmuwe\_csx}{cpmuwe\_total\_csx}
\pointer{cpmuwe}{cnv_cpmre_mono}{\BII}{\Dynamic}{cnv\_cpmre}{cnv\_cpmre\_mono}
\pointer{cpmuwe}{cnv_cpmuwe_mono}{\BII}{\Dynamic}{cnv\_cpmuwe}{cnv\_cpmuwe\_mono}

\pointer{cpes}{cnv_R_cpmuwe_dec}{\BII}{\Dynamic}{cnv\_cpmuwe\_cpmre}{cnv\_R\_cpmuwe\_dec}
\pointer{cpes}{cnv_cpes_dec}{\BII}{\Dynamic}{cnv\_preserve\_cpes}{cnv\_cpes\_dec}
\pointer{cpes}{cnv_dec}{\BII}{\Dynamic}{cnv\_eval}{cnv\_dec}

\pointer{types}{nta_fwd_cnv_sn}{\BII}{\Dynamic}{nta}{nta\_fwd\_cnv\_sn}
\pointer{types}{nta_fwd_cnv_dx}{\BII}{\Dynamic}{nta}{nta\_fwd\_cnv\_dx}
\pointer{types}{nta_fwd_aaa}{\BII}{\Dynamic}{nta\_aaa}{nta\_fwd\_aaa}
\pointer{types}{nta_abst_predicative}{\BII}{\Dynamic}{nta\_aaa}{nta\_abst\_predicative}
\pointer{types}{cnv_cpms_nta}{\BII}{\Dynamic}{nta\_preserve}{cnv\_cpms\_nta}
\pointer{types}{cnv_nta_sn}{\BII}{\Dynamic}{nta\_preserve}{cnv\_nta\_sn}
\pointer{types}{nta_cprs_conf}{\BII}{\Dynamic}{nta\_preserve}{nta\_cprs\_conf}
\pointer{types}{nta_mono}{\BII}{\Dynamic}{nta\_preserve}{nta\_mono}
\pointer{types}{nta_typecheck_dec}{\BII}{\Dynamic}{nta\_eval}{nta\_typecheck\_dec}
\pointer{types}{nta_inference_dec}{\BII}{\Dynamic}{nta\_eval}{nta\_inference\_dec}

\pointer{acle}{cnv_fwd_cpms_abst_dx_le}{\BII}{\Dynamic}{cnv\_aaa}{cnv\_fwd\_cpms\_abst\_dx\_le}
\pointer{acle}{cnv_acle_trans}{\BII}{\Dynamic}{cnv\_acle}{cnv\_acle\_trans}

\pointer{cpts}{cpts_cpms_conf_eq}{\BII}{\Dynamic}{cnv\_cpts}{cpts\_cpms\_conf\_eq}
\pointer{cpts}{cpts_cprs_trans}{\BII}{\RTComp}{cpts\_cpms}{cpts\_cprs\_trans}

\end{itemize}

\end{document}